%% file: main.tex
\documentclass[11pt]{article}

\usepackage[margin=1in]{geometry}

\usepackage[utf8]{inputenc}
\usepackage[T1]{fontenc}  
\usepackage{multirow}
\usepackage{amssymb,amsmath,amsthm,amsfonts}
\usepackage{bm}
\usepackage{textgreek}
\usepackage{mathtools}
\usepackage{enumitem}
\usepackage{authblk}
\usepackage{graphicx}
\usepackage[font=small]{caption}
\usepackage[labelformat=simple]{subcaption}
\usepackage{float}
\usepackage[linesnumbered,ruled,vlined,boxed,algo2e]{algorithm2e}
\usepackage{physics}
\usepackage{footnote}
\usepackage{xcolor}
\usepackage{mathrsfs}
\usepackage{bbm}
\usepackage{makecell}

\usepackage[colorlinks,linkcolor=blue,citecolor=blue,anchorcolor=blue]{hyperref}
\usepackage{tikz}
\usetikzlibrary{trees}
\usetikzlibrary{calc}

\newtheorem{theorem}{Theorem}
\newtheorem{definition}{Definition}
\newtheorem{lemma}{Lemma}
\newtheorem{proposition}{Proposition}

\newtheorem{corollary}{Corollary}

\newtheorem{problem}{Problem}

\input{macros}

\title{Quantum speedups for stochastic optimization}

\author{Aaron Sidford and Chenyi Zhang\\
	\texttt{\{\href{mailto:sidford@stanford.edu}{sidford},%
		\href{mailto:chenyiz@stanford.edu}{chenyiz}\}@stanford.edu}}

\date{}

\begin{document}

\maketitle

\begin{abstract}
We consider the problem of minimizing a continuous function given given access to a natural quantum generalization of a stochastic gradient oracle. We provide two new methods for the special case of minimizing a Lipschitz convex function. Each method obtains a dimension versus accuracy trade-off which is provably unachievable classically and we prove that one method is asymptotically optimal in low-dimensional settings. Additionally, we provide quantum algorithms for computing a critical point of a smooth non-convex function at rates not known to be achievable classically. To obtain these results we build upon the quantum multivariate mean estimation result of Cornelissen et al.~\cite{cornelissen2022near} and provide a general quantum variance reduction technique of independent interest. 
\end{abstract}

\setcounter{tocdepth}{2}

\setcounter{tocdepth}{2}

\tableofcontents

\input{sec-intro}

\input{sec-var-red}


\section{Quantum accelerated stochastic approximation}\label{sec:Q-AC-SA}

In this section, we present our  $\tO(d^{5/8}\epsilon^{-3/2})$ query quantum algorithm for \prob{SCO}. Our approach builds upon the framework proposed by Duchi et al.~\cite{duchi2012randomized,bubeck2019complexity}, which involves performing a Gaussian convolution on the objective function $f$ and then optimizing the resulting smooth convoluted function. Compared to their algorithm, our algorithm differs by replacing the variance reduction step by our quantum variance reduction technique (\algo{unbiased-Q}).

As with a variety of prior work on parallel and private SCO~\cite{duchi2012randomized,gasnikov2018global,bubeck2019complexity}, we consider the smooth function $F_r$ that is the Gaussian convolution of the objective function $f$:
\begin{align}\label{eqn:gaussian-convolution}
F_r(\x)\coloneqq\int_{\R^d}\gamma_r(\y)f(\x-\y)\d \y,\quad\text{where} \quad\gamma_r(\y)\coloneqq\frac{1}{(\sqrt{2\pi}r)^d}\exp\Big(-\frac{\|\y\|^2}{2r^2}\Big).
\end{align}
As shown in \lem{gaussian-smoothing-gradient-smoothness}, when the radius $r$ of the convolution is sufficiently small, $F_r$ closely approximates $f$ pointwise. Consequently, to find an 
$\epsilon$-optimal point of $f$ it suffices to find an $\epsilon/ 4$-optimal point of $F_r$ for $r=\frac{\epsilon}{4\sqrt{d}L}$. Moreover, \lem{gaussian-smoothing-gradient-smoothness} shows that the stochastic gradient $\g_F$ of $F_r$ can be defined and obtained based on the stochastic gradient $\g$ of $f$ as follows:
\begin{align}\label{eqn:g_F-defn}
\g_F(\x)=\g(\x-\y),\qquad\y\sim\gamma_r,
\end{align}
which satisfies
\begin{align*}
\underset{\y\sim\gamma_r}{\mathbb{E}}\g_F(\x)=\nabla F_r(\x)
\quad\text{and}\quad
\|\g_F(\x)\|\leq L,\quad\forall\x.
\end{align*}
Hence,
\begin{align*}
\|\nabla F_r(\x)\|\leq\int_{\y\sim\gamma_r}\|\g_F(\x)\|\d\y\leq L
\end{align*}
indicating that $F_r$ is also $L$-Lipschitz.

To optimize this smooth convex function $F_r$, we leverage the accelerated stochastic approximation (AC-SA) algorithm introduced in~\cite{lan2012optimal}, which applys an accelerated proximal descent method on the objective function using unbiased estimates of gradients. Our algorithm given in \algo{Q-AC-SA} is a specialization of the AC-SA algorithm, where we implement those unbiased estimates of gradients using quantum variance reduction (\algo{unbiased-Q}). In the classical setting, one query to the stochastic gradient $\g_F(\x)$ of $F$ can be implemented by 
a random sampling a vector $\y\in\R^d$ from the Gaussian $\gamma_r$ followed by a query to the SGO $\mathcal{C}_{\g}$ (\defn{Cg}) at $\x-\y$. Similarly, we can show that one query to a QSGO of $F_r$ (\defn{Og}) can also be implemented by one query to the QSGO of $f$. The subsequent theorem presents the query complexity of \algo{Q-AC-SA}. 
\begin{algorithm2e}
	\caption{Quantum accelerated stochastic approximation (Q-AC-SA)}
	\label{algo:Q-AC-SA}
	\LinesNumbered
	\DontPrintSemicolon
    \KwInput{Function $f\colon\R^d\to\R$, precision $\epsilon$}
    \KwParameter{Domain Size $R$, total iteration budget $\mathcal{T}=\frac{4d^{1/4}LR}{\epsilon}$, target variance $\hat{\sigma}=\frac{d^{1/8}}{8}\sqrt{\frac{L\epsilon}{R}}$, convolution radius $r=\frac{\epsilon}{4\sqrt{d}L}$, $\gamma=\frac{R\sqrt{6\ell_F}}{(\mathcal{T}+2)^{3/2}\hat{\sigma}}$}
    \KwOutput{an $\epsilon$-optimal point of $F$}
    Denote $F_r(\x)\coloneqq\int_{\R^d}\gamma_r(\y)f(\x-\y)\d \y$ as in \eqn{gaussian-convolution}\label{lin:convolution}\;
    Set $\x_1\leftarrow \0$, $\xag_1\leftarrow\x_1$\;
    \For{$t=1,2,\ldots,\mathcal{T}$}{
        $\beta_t\leftarrow\frac{t+1}{2}$, $\gamma_t\leftarrow\frac{t+1}{2}\gamma$\;
        $\xmd_t\leftarrow\beta_t^{-1}\x_t+(1-\beta_t^{-1})\xag_t$\;
        Call \algo{unbiased-Q} for an unbiased estimate $\tilde{\g}_t$ of $\nabla F_r(\xmd_t)$ with variance at most $\hat{\sigma}^2$\label{lin:convex-QVR}\;
        $\x_{t+1}\leftarrow \underset{\z\in\B_R(\0)}{\arg\min}\left\{\gamma_t\<\tilde{\g}_t,\z-\xmd_t\>+L\|\xmd_t-\z\|^2/(2r)\right\}$\label{lin:proximal}\;
        $\xag_{t+1}=\beta_t^{-1}\x_{t+1}+(1-\beta_t^{-1})\xag_t$
    }
        \Return $\xag_{\mathcal{T}+1}$\;
\end{algorithm2e}
\begin{theorem}[Formal version of \thm{SCO-informal}, Part 1]\label{thm:Q-AC-SA}
\algo{Q-AC-SA} solves \prob{SCO} using an expected $\tO(d^{5/8}(LR / \epsilon)^{3/2})$ queries.
\end{theorem} 

To prove \thm{Q-AC-SA}, we first state a result from \cite{bubeck2019complexity} which bounds properties of $F$ defined in \eqn{gaussian-convolution}. The lemma states that, provided the radius $r$ of the convolution is sufficiently small, $F$ is pointwise close to $f$ and hence finding an $\mO(\epsilon)$-minimum of $f$ is equivalent to finding an $\epsilon$-minimum of $F$. Further the lemma shows that $F$ is Lipschitz and smooth. 

\begin{lemma}[{\cite[Lemma 8]{bubeck2019complexity}}]\label{lem:gaussian-smoothing-gradient-smoothness}
For any $L$-Lipschitz convex function $f\colon\R^d\to\R$, its Gaussian convolution $g$ defined in Eq.~\eqn{gaussian-convolution} is convex, $L$-Lipschitz, and satisfies
\begin{align*}
|f(\x)-g(\x)|\leq\sqrt{d}\cdot Lr
\quad\text{ 
and
}\quad
\nabla^2g(\x)\leq(L/r)\cdot \mathbb{I}_d,\quad\forall\x\in\R^d,
\end{align*}
where $\mathbb{I}_d$ is the $d$-dimensional identity matrix.
\end{lemma}

Prior to proving \thm{Q-AC-SA}, we show that one query to a QSGO $O_{\g_F}$ (\defn{Og}) of $F_r$ satisfying
\begin{align}\label{eqn:O_gF-defn}
O_{\g_F}\ket{\x}\otimes\ket{0}\otimes\ket{0}\to\ket{\x}\otimes\int_{\v\in\R^d}\sqrt{p_{F,\x}(\v)\d \v}\ket{\v}\otimes\ket{\mathrm{garbage}(\v)},
\end{align}
can be implemented by one query to $O_\g$ of $f$ defined in \defn{Og}.

\begin{lemma}\label{lem:O_gF}
The QSGO $O_{\g_F}$ of $F_r$ defined in \eqn{O_gF-defn} can be implemented with one query to the QSGO $O_\g$ of $f$.
\end{lemma}
\begin{proof}
We first prepare the following quantum state
\begin{align*}
\ket{\psi}=\ket{\x}\otimes\left(\int_{\y\in\R^d}\sqrt{\gamma_r(\y)\d\y}\ket{\x-\y}\right)\otimes\ket{0},\quad\forall\x\in\R^d.
\end{align*}
Applying $O_{\g}$ to the last two registers yields
\begin{align}
(\mathbb{I}\otimes O_\g)\ket{\psi}
&=\ket{\x}\otimes\left(\int_{\y\in\R^d}\sqrt{\gamma_r(\y)\d\y}\ket{\x-\y}\otimes\int_{\v\in\R^d}\sqrt{p_{f,\x-\y}(\v)\d\v}\ket{\v}\right)\nonumber\\
&=\ket{\x}\otimes\int_{\y,\v\in\R^d}\sqrt{\gamma_r(\y)p_{f,\x-\y}(\v)\d\y\d\v}\ket{\x-\y}\ket{\v}.\label{eqn:oracle-middle-step}
\end{align}
Given that
\begin{align*}
p_{F,\x}(\v)=\int_{\y\in\R^d}\gamma_r(\y)p_{f,\x-\y}(\v)\d\y,\qquad\forall\v\in\R^d,
\end{align*}
if we measure the last register, the probability density function of the outcome would be exactly $p_{F,\x}$. Hence, the quantum state $(\mathbb{I}\otimes O_\g)\ket{\psi}$ in \eqn{oracle-middle-step} can also be written as
\begin{align*}
(\mathbb{I}\otimes O_\g)\ket{\psi}=\ket{\x}\otimes\int_{\v\in\R^d}\sqrt{p_{F,\x}(\v)\d \v}\ket{\text{garbage}(\v)}\otimes\ket{\v}.
\end{align*}
By swapping the last two quantum registers, we can obtain the desired output state of $O_{\g_F}$.
\end{proof}

The following result from \cite{lan2012optimal} bounds the rate at which \algo{Q-AC-SA} decreases the function error of $F_r$. Note that validity of this result relies solely on the fact that, at \lin{convex-QVR}, the variance of the unbiased gradient estimate $\tilde{\g}_t$ does not exceed $\hat{\sigma}^2$, irrespective of its implementation.

\begin{lemma}
\label{lem:AC-SA-convergence}
The output $\xag_{\mathcal{T}+1}$ of \algo{Q-AC-SA} satisfies 
\begin{align*}
\E[F_r(\xag_{\mathcal{T}+1})-F_r^*]\leq\frac{4L R^2}{r\mathcal{T}(\mathcal{T}+2)}+\frac{4R\hat{\sigma}}{\sqrt{\mathcal{T}}},
\end{align*}
where $F_r$ is the convoluted function defined in~\lin{convolution} of \algo{Q-AC-SA} and $F_r^*$ is its minimum.
\end{lemma}

\begin{proof}
This lemma follows from \cite[Corollary 1]{lan2012optimal}, which shows that
\begin{align*}
\E[F_r(\xag_{\mathcal{T}+1})-F_r^*]\leq\frac{4\ell_F R^2}{\mathcal{T}(\mathcal{T}+2)}+\frac{4R\hat{\sigma}}{\sqrt{\mathcal{T}}},
\end{align*}
where $\ell_F$ is the Lipschitz parameter of $F_r$, which equals $L/r$ by \lem{gaussian-smoothing-gradient-smoothness}.
\end{proof}
Equipped with above results, we are now ready to present the proof of \thm{Q-AC-SA}. 

\begin{proof}[Proof of \thm{Q-AC-SA}]
By \lem{AC-SA-convergence}, the output $\x_{\out}$ of \algo{Q-AC-SA} satisfies
\begin{align}
\E[F(\x_{\out})-F^*]\leq\frac{4\ell_F R^2}{r\mathcal{T}(\mathcal{T}+2)}+\frac{4R\hat{\sigma}}{\sqrt{\mathcal{T}}}\leq \frac{\epsilon}{2}.
\end{align}
Moreover, by \lem{gaussian-smoothing-gradient-smoothness} we can derive that
\begin{align*}
\E[f(\x_{\out})-f^*]\leq \E[f(\x_{\out})-F(\x_{\out})]+(f^*-F^*)+\E[F(\x_{\out})-F^*]=\epsilon.
\end{align*}

Finally we bound the expected number of queries to $O_\g$ used in the algorithm. Since $\|\g_F(\x)\|\leq L$ for any $\x\in\R^d$, by \thm{unbiased-estimator} we know that \algo{unbiased-Q} can output an estimate $\tilde{\g}_F(\xmd_t)$ of $\nabla F(\xmd_t)$ with variance at most $\hat{\sigma}^2$ using $\tO(L\sqrt{d}/\hat{\sigma})$ queries to the oracle $O_F$, and hence the same number of queries to $O_\g$ by \lem{O_gF}. Then, the total number of queries to $O_\g$ equals
\begin{align*}
\mathcal{T}\cdot\tO\left(\frac{L\sqrt{d}}{\hat{\sigma}}\right)=\frac{d^{1/4}LR}{\epsilon}\cdot\tO\left(\frac{L\sqrt{d}}{d^{1/8}}\cdot\sqrt{\frac{R}{L\epsilon}}\right)=\tO\left(d^{5/8}\cdot\left(\frac{LR}{\epsilon}\right)^{3/2}\right).
\end{align*}
\end{proof}


\section{Quantum stochastic cutting plane method (Q-SCP)}\label{sec:QSCPM}

In this section, we develop our $\tO(d^{3/2}LR/\epsilon)$ query algorithm for \prob{SCO} which is based on a stochastic version of the cutting plane method. We introduce the key properties and related concepts of cutting plane methods, and then provide a procedure for efficiently post-processing the outcomes obtained from the stochastic cutting plane method using quantum variance reduction (\algo{unbiased-Q}). Then, we analyze the overall query complexity for solving \prob{SCO}. Technically, it is possible to obtain the results of this section using the quantum mean estimation routine of \cite{cornelissen2022near}, rather then quantum variance reduction, however using quantum variance reduction facilitates our presentation. 

\subsection{Properties of the stochastic cutting plane method}

We begin by introducing some notation and concepts on cutting plane methods. Cutting plane methods solve the \textit{feasibility problem} defined as follows. Note that this problem is slightly easier to solve then the one in, e.g., \cite{jiang2020improved}, however it is simple suffices for our purposes.

\begin{problem}[Feasibility Problem]\label{prob:feasibility}
We are given query access to a separation oracle for a set $K\subset\R^d$ such that on query $\x\in\R^d$ the oracle outputs a vector $\vect{c}$ and either $\vect{c} = \0$, in which case $\x \in K$, or $\vect{c} \neq \0$, in which case $H\coloneqq\{\vect{z}\colon\vect{c}^\top\vect{z}\leq\vect{c}^\top \x\}\supset K$. The goal is to query a point $\x\in K$.
\end{problem}

\cite{jiang2020improved} showed that \prob{feasibility} can be solved by cutting plane method using $\mO(d\log(dR/r))$ queries to a separation oracle where $R$ and $r$ denote bounds on $K$.

\begin{lemma}[{\cite[Theorem 1.1]{jiang2020improved}}]\label{lem:classical-cutting-plane}
There is a cutting plane method which solves \prob{feasibility} using at most $C\cdot d\log(dR/r)$ queries for some constant $C$, given that the set $K$ is contained in the ball of radius $R$ centered at the origin and it contains a ball of radius $r$. 
\end{lemma}

\cite{nemirovski1994efficient,lee2015faster} demonstrated that, running cutting plane method on a convex function $f$ with the separation oracle being its gradient yields a sequence of points where at least one of them is an $\epsilon$-optimal point of $f$. This follows from the fact that there exists a ball of radius $\mO(\epsilon)$ around $\x^*$ such that every point in this ball is $\epsilon$-optimal. In the stochastic setting, although we cannot access the precise gradient, we show that it suffices to use use an $\mO(\epsilon/R)$-approximate gradient oracle of $f$ (defined below) as the separation oracle. 

\begin{definition}[$\delta$-approximate gradient oracle ($\delta$-AGO)]\label{defn:AGO}
For $f\colon\R^d\to\R$, its \emph{$\delta$-approximate gradient oracle,
$\delta$-AGO,} is defined as a random function that when queried at $\x$, returns a vector $\tilde{\g}(\x)$ that satisfies $\|\tilde{\g}(\x)-\nabla f(\x)\|\leq\delta$.
\end{definition}

We show that a $\delta$-AGO can be efficiently implemented using our quantum variance reduction algorithm. To obtain this result we can also use \cite[Claim 1]{kelner2023semi} by changing some parameters. Nevertheless, compared to their deterministic implementation, our approach has a lower expected time complexity (though is randomized).

\begin{lemma}\label{lem:QVR-bounded-error}
For any $\delta,\xi\geq 0$, with success probability at least $1-\xi$, the $\delta$-AGO of a Lipschitz function $f$ can be implemented using $\mO(\log(1/\xi))$ calls to \algo{unbiased-Q} with in total $\mO(L\sqrt{d}\log (1/\xi)/\delta)$ queries to the QSGO.
\end{lemma}

\begin{proof}
We implement the $\delta$-AGO using the following procedure. First, we obtain $k=\log(1/\xi)+10$ independent unbiased estimates $\g_1,\ldots,\g_k$ of $\nabla f(\x)$ with variance $\delta^2/16$ by \algo{unbiased-Q}, using
\begin{align*}
k\cdot\frac{L\sqrt{d}}{\delta}=\mO(L\sqrt{d}\log (1/\xi)/\delta)
\end{align*}
queries in total to the QSGO $O_\g$, by \thm{unbiased-estimator}. Note that any individual sample $\g_i$ satisfies
\begin{align*}
\Pr[\|\g_i-\nabla f(\x)\|\geq\frac{\delta}{2}]\leq\frac{1}{4}.
\end{align*}
Hence, if we let $S$ denote the subset of $[k]$ where each $i \in S$
have $\ell_2$ distance to $\nabla f(\x)$ at most $\delta/2$, then by Hoeffding's inequality, $|S| \geq 2k/3$ with probability at least
\begin{align*}
1-\exp\left(-\frac{2(k/3-k/4)^2}{k}\right)\geq 1-\xi\,.
\end{align*}

Observe that for any sample $\g_i$ satisfying $\|\g_i-\g_j\|\leq\delta/2$ for at least $2k/3$ of $\g_j$'s, there must exist at least some $\g_{j'}\in S$ such that $\|\g_i-\g_{j'}\|\leq\delta/2$, and therefore $\|\g_i-\nabla f(\x)\|\leq\delta$ by triangle inequality. Consequently, we just need to find and output any such $\g_i$ as an $\delta$-AGO, which is guaranteed to exist given that any sample in $S$ satisfies this property. This can be done in expected $\mO(k)$ time by picking a random estimate and checking whether it is close to at least $2k/3$ of other estimates.
\end{proof}

Next, we establish that we can query an $\epsilon$-optimal point by applying the cutting plane method with the separation oracle being an $\mO(\epsilon/R)$-approximate gradient oracle, as provided by \lem{QVR-bounded-error}.

\begin{proposition}\label{prop:cutting-plane}
For any $0\leq\epsilon\leq LR$, with success probability at least $5/6$ we can obtain $\mathcal{T}=\mO(d\log(dLR/\epsilon))$ points
$\x_1,\ldots,\x_{\mathcal{T}}\in\B_R(\0)$ such that one of the $x_i$ is $\epsilon$-optimal using $\tO\left(d^{3/2}LR/\epsilon\right)$ queries to the QSGO $O_\g$ defined in \defn{Og}.
\end{proposition}

\begin{proof}
Define $K_{\epsilon/2}$ as the set of $\epsilon/2$-optimal points of the 
function $f$, and $K_{\epsilon}$ as the set of $\epsilon$-optimal points of $f$.
We know that $K_{\epsilon/2}$ contains a ball of radius at least $r_K = \epsilon/(2L)$ since for any $\x$ with $\|\x-\x^*\|\leq r_K$ we have
\begin{align}
f(\x)-f(\x^*)\leq L\|\x-\x^*\|\leq\frac{\epsilon}{2}.
\end{align}
We apply the cutting plane method, as described in \lem{QVR-bounded-error}, to query a point in $K_{\epsilon/2}$, which is a subset of the ball $\mathbb{B}_{2R}(\mathbf{0})$. To achieve this, we use an $\epsilon/(10R)$-approximate gradient oracle ($\epsilon/(10R)$-AGO) of $f$, implemented using \lem{QVR-bounded-error}, as the separation oracle for the cutting plane method. Throughout this process, we assume each query to the $\epsilon/(10R)$-AGO is successfully executed, and we will later discuss the error probability associated with this assumption. In this case, we show that any query outside of $K_{\epsilon}$ to the $\epsilon/(10R)$-AGO will be a valid separation oracle for $K_{\epsilon/2}$. In particular, if we ever queried the $\epsilon/(10R)$-AGO at any $\x\in\B_{2R}(\0)\setminus K_{\epsilon}$ with output being $\tilde{\g}$, for any $\y\in K_{\epsilon/2}$ we have
\begin{align*}
\left<\tilde{\g},\y-\x\right>
&\leq\left\<\nabla f(\x),\y-\x\right\>+\|\tilde{\g}-\nabla f(\x)\|\cdot\|\y-\x\|\\
&\leq f(\y)-f(\x)+\|\tilde{\g}-\nabla f(\x)\|\cdot\|\y-\x\|\\
&\leq -\frac{\epsilon}{2}+\frac{\epsilon}{10R}\cdot 4R<0,
\end{align*}
where the second inequality is due to the convexity of $f$, indicating that $\tilde{\g}$ is a valid separation oracle for the set $K_{\epsilon/2}$. Consequently, upon applying \lem{classical-cutting-plane}, we can deduce that after $Cd\log(dR/r_K)$ iterations, at least one of the queries must lie within $K_{\epsilon}$. 

To ensure an overall success probability of at least $5/6$, we employ the union bound, which necessitates that each query to the $\epsilon/(10R)$-AGO be implemented with a failure probability of no more than $(6Cd\log(R/r_K))^{-1}$, which by \lem{QVR-bounded-error} requires  
\begin{align*}
\mO\left(\frac{10LR\sqrt{d}}{\epsilon}\log (6Cd\log(dR/r_K))\right)=\tO\left(\frac{LR\sqrt{d}}{\epsilon}\right)
\end{align*}
queries to the QSGO $O_{\g}$, and the overall query complexity equals
\begin{align*}
\tO\left(\frac{LR\sqrt{d}}{\epsilon}\right)\cdot Cd\log(dR/r_K)=\tO\left(\frac{d^{3/2}LR}{\epsilon}\right).
\end{align*}
\end{proof}

\subsection{Stochastic approximate optimal finding among finite number of points}

After applying \prop{cutting-plane}, it is not clear which query $\x_i$ is an $\mO(\epsilon)$-optimal point. This difficulty arises because we lack access to the function value of $f$, which sets our problem apart from the feasibility problem discussed in~\cite{jiang2020improved}, where there is a clear indication when a query successfully lies within the feasible region. Consequently, 
we next focus on identifying the optimal solution within the finite set $\Gamma$ of points using access to the QSGO. We conceptualize this task as the \textit{approximately best point} problem, formulated as follows.

\begin{problem}[Approximately best point]\label{prob:finite-optimal-finding}
For a $L$-Lipschitz convex function $f\colon\R^d\to\R$ and $\mathcal{T}$ points $\x_1,\ldots,\x_{\mathcal{T}}\in\B_R(\0)$, find a convex combination $\hat{\x}\in\R^d$ of the points satisfying
\begin{align*}
f(\hat{\x})\leq\min_{t \in \mathcal{T}} f(\x_t)+\epsilon.
\end{align*}
\end{problem}

In this work, we develop an algorithm that solves \prob{finite-optimal-finding} by making pairwise comparisons in a hierarchical order, where each pairwise comparison is computed by running binary search on the segment along the segment between the two points. Note that in \lin{projected} of \algo{sls} the estimation step is carried out in a projected space along the vector $\hat{\mathbf{e}}$. This is motivated by the fact that we are specifically interested in the information of the gradient $\nabla f$ within this projected space, which effectively reduces to a one-dimensional variable. As a result, the mean estimation can be performed without introducing an additional factor of $\sqrt{d}$ in the query complexity as given in \thm{unbiased-estimator}.

\begin{algorithm2e}
	\caption{
 \texttt{StochasticLineSearch}$\left(\y_{l0},\y_{r0}, \epsilon'\right)$}
	\label{algo:sls}
	\LinesNumbered
	\DontPrintSemicolon
	\KwInput{Endpoints $\y_l,\y_r\in\B_R(\0)$, accuracy $\epsilon'$}
	\KwOutput{$\hat{\y}$ such that $f(\hat{\y})\leq\min_{\lambda\in[0,1]} f\big(\lambda\y_l+(1-\lambda)\y_r\big)+\epsilon'$.}
	$\y_l\leftarrow\y_{l0}$, $\y_r\leftarrow\y_{r0}$\;
	\lIf{$\y_l=\y_r$}{\Return $\y_{l}$}
	\Else{
	$\hat{\vect{e}}\leftarrow\frac{\y_r-\y_l}{\|\y_r-\y_l\|}$\; 
    \Repeat{$\|\y_r-\y_l\|\leq \epsilon'/L$}{
    $\y_m\leftarrow(\y_l+\y_r)/2$\;
        Obtain an estimate $\tilde{g}_{\hat{\e}}(\y_m)$ of $\nabla f(\y_m)^{\top}\hat{\vect{e}}$ with error at most $\epsilon'/(4R)$ \label{lin:projected}\; 
        \lIf{$|\tilde{g}_{\hat{\e}}(\y_m)|\leq\epsilon'/(4R)$}{
        \Return $\hat{\y}\leftarrow\y_m$
        }
        \Else{
        \leIf{$\tilde{g}_{\hat{\e}}(\y_m)>0$}{
        $\y_r\leftarrow\y_m$
        }{
        $\y_l\leftarrow\y_m$
        }
        }
    }
    \Return $\hat{\y}\leftarrow\y_l$
    }
\end{algorithm2e}

\begin{lemma}\label{lem:sls}
For any $\y_{l0},\y_{r0}\in\B_R(\0)$ and any $\epsilon'>0$, with success probability at least $1/(6\mathcal{T})$ \algo{sls} returns a point $\hat{\y}\in\R^d$ satisfying
\begin{align*}
f(\hat{\y})\leq\min_{\lambda\in[0,1]} f\big(\lambda\y_{l0}+(1-\lambda)\y_{r0}\big)+\epsilon'
\end{align*}
using $\tO\left(RL/\epsilon'\right)$ queries to an $L$-bounded QSGO $O_\g$ defined in \defn{Og}.
\end{lemma}

\begin{proof}
Observe that in each iteration the value $\|\y_r-\y_l\|$ is decreased by at least $1/2$. Hence, the total number of iterations is at most 
\begin{align*}
\log_2\left(\frac{\|\y_{r0}-\y_{l0}\|}{\epsilon'/L}\right)\leq \log_2\left(\frac{2RL}{\epsilon'}\right).
\end{align*}
Within each iteration, we need to estimate the value $\nabla f(\y_m)^\top\hat{\e}$, which is the component of $\nabla f(\y_m)$ along $\hat{\vect{e}}$. Note that this is essentially a univariate mean estimation problem, given that $\hat{\e}$ is fixed in this iteration we can define an 1-dimensional random variable $\nabla f(\y_m)^\top\hat{\e}$ whose variance is at most
\begin{align*}
\mathrm{Var}[\nabla f(\y_m)^\top\hat{\e}]\leq \E|\nabla f(\y_m)^\top\hat{\e}|^2\leq\E\|\nabla f(\y_m)\|^2\leq L^2.
\end{align*}
Then by \lem{QVR-bounded-error}, with success probability at least
\begin{align}
1-\left(6\mathcal{T}\log_2\left(\frac{2RL}{\epsilon'}\right)\right)^{-1},
\end{align}
an estimate of $\nabla f(\y_m)^\top\hat{\e}$ with error at most $\epsilon'/(4R)$ can be obtained using
\begin{align*}
\frac{L}{\epsilon'/(4R)}\cdot\log\left(6\mathcal{T}\log_2\left(\frac{2RL}{\epsilon'}\right)\right)=\tO\left(\frac{RL}{\epsilon'}\right)
\end{align*}
queries to the quantum stochastic gradient oracle $O_{\g}$, given that we can prepare the following quantum oracle
\begin{align*}
O_{\g^\top\hat{\e}}\ket{\x}\otimes\ket{0}\to\ket{\x}\otimes\int_{\v\in\R^d}\sqrt{p_{f,\x}(\v)\d \v}\ket{\v^\top\hat{\e}}\otimes\ket{\mathrm{garbage}(\v)},
\end{align*}
using one query to $O_\g$. Hence, \algo{sls} uses $\tO(RL/\epsilon')$
queries to the quantum stochastic gradient oracle $O_\g$, with failure probability at most
\begin{align}
\log_2\left(\frac{\|\y_{r0}-\y_{l0}\|}{\epsilon'/L}\right)\cdot \left(6\mathcal{T}\log_2\left(\frac{2RL}{\epsilon'}\right)\right)^{-1}\leq\frac{1}{6\mathcal{T}}
\end{align}
by union bound. Next, we show that the output $\hat{\y}$ of \algo{sls} has a relatively small function value as desired. We denote
\begin{align*}
\lambda^* \defeq \underset{\lambda\in[0,1]}{\arg\min}f(\lambda\y_{l0}+(1-\lambda)\y_{r0})
\quad\text{ and }\quad
\y^* \defeq \lambda^*\y_{l0}+(1-\lambda^*)\y_{r0}.
\end{align*}
If the algorithm terminates at a point $\hat{\y}$ satisfying $|\tilde{g}_{\hat{\e}}(\hat{\y})|\leq\epsilon'/(4R)$, by convexity and the Cauchy Schwarz inequality we have
\begin{align*}
f(\hat{\y})-f(\y^*)
\leq |\nabla f(\hat{\y})^\top\hat{\vect{e}}|\cdot\|\hat{\y}-\y^*\|\leq \frac{\epsilon'}{2R}\cdot 2R\leq\epsilon'.
\end{align*}
Otherwise, by induction we can demonstrate that in each iteration of the algorithm, $\y^*$ resides within the segment bounded by $\y_l$ and $\y_r$. Assume that this assertion holds true for the $t$-th iteration. If $\tilde{g}_{\hat{\e}}(\y_m)>\epsilon'/(4R)$, we have
\begin{align*}
\<\nabla f(\y_m),\y_l-\y_m\>
&=-\<\nabla f(\y_m),\hat{\e}\>\cdot\|\y_l-\y_m\|\\
&\leq -\left(\tilde{g}_{\hat{e}}(\y_m)-\frac{\epsilon'}{4R}\right)\|\y_l-\y_m\|< 0,
\end{align*}
indicating that the value of $f$ will decrease along the direction $\y_l-\y_m$. Hence, $\y^*$ lies in the segment between $\y_l$ and $\y_m$ of the $t$-th iteration, or equivalently, the segment between $\y_l$ and $\y_m$ of the $(t+1)$-th iteration, given that $f$ is convex. A symmetric argument applies in the case of $\tilde{g}_{\hat{\e}}(\y_m)<\epsilon'/(4R)$. Then, we have
\begin{align*}
\|\y_l-\y^*\|\leq\|\y_l-\y_r\|
\end{align*}
for every iteration, which leads to
\begin{align*}
\|\hat{\y}-\y^*\|\leq \frac{\epsilon'}{L}
\end{align*}
when the algorithm terminates. Hence,
\begin{align*}
f(\hat{\y})-f(\y^*)\leq L\cdot\|\hat{\y}-\y^*\|=\epsilon'
\,
\end{align*}
considering that $f$ is $L$-Lipschitz.
\end{proof}

Using \algo{sls} as a subroutine, we develop \algo{finite-optimal-finding} that solves \prob{finite-optimal-finding}.

\begin{algorithm2e}
	\caption{Stochastic Approximately best point among finite number of points
 }
	\label{algo:finite-optimal-finding}
	\LinesNumbered
	\DontPrintSemicolon
	\KwInput{A set of points $\{\x_1,\ldots,\x_{\mathcal{T}}\}\subset\B_R(\0)$ where $\mathcal{T}$ is a power of 2, accuracy $\epsilon$}
	\KwOutput{$\hat{\x}$ such that $f(\hat{\x})\leq\min_i f(\x_i)+\epsilon$.}
	$\y_{0,i}\leftarrow\x_i$ for all $i\in[\mathcal{T}]$\;
	\For{$\tau = 1, \ldots, \log_2 \mathcal{T}$}
	{
        \lFor{$j = 1,\ldots,\mathcal{T}/2^{\tau}$}
        {
            $\y_{\tau,j}=$\texttt{StochasticLineSearch}$\left(\y_{\tau-1,2j-1},\y_{\tau-1,2j},\epsilon/\log_2\mathcal{T}\right)$\label{lin:sls}
        }
	}
    \Return $\y_{\log\mathcal{T},1}$\;
\end{algorithm2e}

\begin{proposition}\label{prop:AOF-query-complexity}
For any accuracy parameter $\epsilon>0$, with success probability at least $5/6$ \algo{finite-optimal-finding} solves \prob{finite-optimal-finding} using $\tO\left(RL\mathcal{T}/\epsilon\right)$ queries to an $L$-bounded QSGO $O_\g$ defined in \defn{Og}.
\end{proposition}

\begin{proof}
We start by demonstrating that without loss of generality we can assume $\mathcal{T}$ to be a power of 2. If $\mathcal{T}$ is not already a power of 2, we can simply augment the set with at most $\mathcal{T}$ points, each being $\x_1$, without changing the algorithm's output.

Next, we show that \algo{finite-optimal-finding} makes $\mathcal{T}-1$ calls to the $\sls(\cdot)$ subroutine (\algo{sls}) with an accuracy of $\epsilon'=\epsilon/\log_2\mathcal{T}$ to solve the \prob{finite-optimal-finding} problem. The algorithm exhibits a hierarchical structure, as illustrated in \fig{finite-optimal-finding}, where each non-leaf node invokes the $\sls(\cdot)$ subroutine once. Consequently, the total number of calls is $\mathcal{T}-1$.

\begin{figure}[htbp]
    \centering
    \tikzset{global scale/.style={
    scale=#1,
    every node/.append style={scale=#1}
  }
}
    \begin{tikzpicture}[level distance=2cm,
  level 1/.style={sibling distance=8cm},
  level 2/.style={sibling distance=4cm},
  level 3/.style={sibling distance=2cm},
  nodes={draw, circle}, minimum size=4em,global scale=0.85]
  
  \node {$\y_{t,1}$}
    child {node {$\y_{t-1,1}$}
      child {node(1) {$\y_{1,1}$}
        child[solid]{ node(3) {$\x_1$}
        edge from parent node[left,xshift=-2.5,draw=none]{}}
        child[solid]{ node {$\x_{\frac{\mathcal{T}}{4}}$}
        edge from parent node[right,xshift=2.5,draw=none]{}}
      edge from parent [dashed] node[left,xshift=-2.5,draw=none]{}
      }
      child {node {$\y_{1,\frac{\mathcal{T}}{4}}$} 
        child[solid]{ node {$\x_{\frac{\mathcal{T}}{2}-1}$}
        edge from parent node[left,xshift=-2.5,draw=none]{}}
        child[solid]{ node {$\x_{\frac{\mathcal{T}}{2}}$}
        edge from parent node[right,xshift=2.5,draw=none]{}}
      edge from parent  [dashed]node[right,xshift=2.5,draw=none]{}}
      edge from parent node[left,xshift=-2.5,draw=none]{}
    }
    child {node {$\y_{t-1,2}$}
    child {node {$\y_{1,\frac{\mathcal{T}}{4}+1}$} 
        child[solid]{ node {$\x_{\frac{\mathcal{T}}{2}+1}$}
        edge from parent node[left,xshift=-2.5,draw=none]{}}
        child[solid]{ node {$\x_{\frac{\mathcal{T}}{2}+2}$}
        edge from parent node[right,xshift=2.5,draw=none]{}}
    edge from parent [dashed] node[left,xshift=-2.5,draw=none]{}}
      child {node(2) {$\y_{1,\frac{\mathcal{T}}{2}}$} 
        child[solid]{ node {$\x_{\mathcal{T}-1}$}
        edge from parent node[left,xshift=-2.5,draw=none]{}}
        child[solid]{ node(4) {$\x_{\mathcal{T}}$}
        edge from parent   node[right,xshift=2.5,draw=none]{}}
      edge from parent  [dashed]}
      edge from parent
    };
\end{tikzpicture}

    \caption{The hierarchical structure of \algo{finite-optimal-finding}.}
    \label{fig:finite-optimal-finding}
\end{figure}

By \lem{sls}, each call to $\sls(\cdot)$ is executed successfully with probability at least $1-1/(6\mathcal{T})$. Then by union bound, we can deduce that the likelihood of all calls to $\sls(\cdot)$ being successful is at least $5/6$. In this case, for any node $\y_{\tau,j}$ in the tree with $\tau\in[\log_2\mathcal{T}]$ and $j\in[\mathcal{T}/2^\tau]$, by convexity and the guarantee of \lem{sls} we have 
\begin{align*}
f(\y_{\tau,j})&\leq\min_{\lambda\in[0,1]}f\big(\lambda\y_{\tau-1,2j-1}+(1-\lambda)\y_{\tau-1,2j}\big)+\epsilon/\log_2\mathcal{T}\\
&\leq \min\{f(\y_{\tau-1,2j-1}),f(\y_{\tau-1,2j})\}+\epsilon/\log_2\mathcal{T}.
\end{align*}
Therefore for any leaf node $\x_j$ with $j\in[\mathcal{T}]$, by summing over the path between the root $\y_{t,1}$ and $\x_j$ we have
\begin{align*}
f(\y_{t,1})\leq f(\x_j)+t\cdot\frac{\epsilon}{\log_2\mathcal{T}} = f(\x_j)+\epsilon,
\end{align*}
which leads to 
\begin{align*}
f(\y_{t,1})\leq\min_{j\in[\mathcal{T}]} f(\x_j)+\epsilon,
\end{align*}
indicating $\y_{t,1}$ is a valid solution of \prob{finite-optimal-finding}. Note that there are in total $\mathcal{T}-1$ non-leaf nodes in \fig{finite-optimal-finding}, and the value of each non-leaf node is computed by one call to the subroutine $\sls(\cdot)$. Hence, the total number of calls to $\sls(\cdot)$ equals $\mathcal{T}-1$, and the overall failure probability of \algo{finite-optimal-finding} is at most
\begin{align}
(\mathcal{T}-1)\cdot\frac{1}{6\mathcal{T}}\leq\frac{1}{6}
\end{align}
by union bound, as the failure probability of \algo{sls} is at most $1/(6\mathcal{T})$. Since each call to \algo{sls} takes 
\begin{align*}
\tO\left(\frac{RL}{\epsilon'}\right)=\tO\left(\frac{RL}{\epsilon}\right)
\end{align*}
queries to the quantum stochastic gradient oracle $O_\g$ defined in \eqn{quantum-SG-oracle}, the total number of queries to the quantum stochastic gradient oracle $O_\g$ is then $\tO\left(RL\mathcal{T}/\epsilon\right)$.
\end{proof}

\subsection{Query complexity of quantum stochastic cutting plane method}

Next, we present the main result of this section, which describes the query complexity of solving \prob{SCO} using quantum stochastic cutting plane method.

\begin{corollary}[Formal version of \thm{SCO-informal}, Part 2]\label{cor:QSCPM}
With success probability at least $2/3$, \prob{SCO} can be solved using an expected $\tO\left(d^{3/2}LR/\epsilon\right)$
queries.
\end{corollary}

\begin{proof}
We first run cutting plane method with the separation oracle being an $\frac{\epsilon}{10R}$-approximate gradient of $f$ implemented by \algo{unbiased-Q}, which by \prop{cutting-plane} outputs a set $\Gamma$ of $\mathcal{T}=\tO(d\log(L/\epsilon))$ points containing at least one $\mO(\epsilon)$-optimal point of $f$ using $\tO\left(d^{3/2}LR/\epsilon\right)$ queries to $O_\g$ with success probability at least $5/6$. Then by running \algo{finite-optimal-finding} on $\Gamma$, with success probability at least $5/6$ we can find a point $\hat{\x}\in\R^d$ satisfying
\begin{align*}
f(\hat{\x})\leq\min_{\x\in\Gamma}f(\x)+\epsilon\leq f(\x^*)+\mO(\epsilon),
\end{align*}
which takes $\tO(RL\mathcal{T}/\epsilon)=\tO(dRL/\epsilon)$ queries to $O_{\g}$ by \prop{AOF-query-complexity}. Hence, the overall number of queries to $O_\g$ equals $\tO(d^{3/2}LR/\epsilon)$, and the overall success probability is at least $2/3$.
\end{proof}


\section{Quantum stochastic non-convex optimization}\label{sec:nonconvex}

In this section, we present our quantum algorithms for \prob{nonconvex} in the bounded-variance setting (\sec{bounded-variance-algorithm}) and the mean-squared smoothness setting (\sec{mean-squared-smoothness-algorithm}), respectively, using our quantum variance reduction technique. 

\subsection{Algorithm for the bounded-variance setting}\label{sec:bounded-variance-algorithm}

To solve \prob{nonconvex} in the bounded variance setting, we leverage the randomized SGD method introduced in~\cite{ghadimi2013stochastic}, which is a variant of SGD where the number of iterations is randomized. Our algorithm is a specialization of the randomized stochastic gradient algorithm, where we replace the classical variance reduction step by quantum variance reduction (\algo{unbiased-Q}). The query complexity of our quantum algorithm is given in the following theorem.

\begin{algorithm2e}
	\caption{Quantum randomized SGD (Q-SGD)}
	\label{algo:Q-SGD}
	\LinesNumbered
	\DontPrintSemicolon
    \KwInput{Function $f\colon\R^d\to\R$, precision $\epsilon$, variance $\sigma$, smoothness $\ell$}
    \KwParameter{$\hat{\sigma}=\epsilon/3$, total iteration budget $\mathcal{T}=12\Delta\ell\epsilon^{-2}$}
    \KwOutput{$\epsilon$-critical point of $f$}
    Uniformly randomly select $N$ from $1,\ldots, \mathcal{T}$\;
    Set $\x_0\leftarrow \0$\;
    \For{$t=0,1,2,\ldots,N-1$}{
        Call \algo{unbiased-Q} for an unbiased estimate $\tilde{\g}_t$ of $\nabla f(\x_t)$ with variance at most $\hat{\sigma}^2$\;
        $\x_{t+1}\leftarrow\x_t-\tilde{\g}_t/\ell$
    }
    \Return $\x_N$
\end{algorithm2e}

\begin{theorem}[Formal version of \thm{Q-SPIDER-informal}, bounded variance setting]\label{thm:Q-SGD} For any $\epsilon>0$, \algo{Q-SGD} solves \prob{nonconvex} in the bounded variance setting using an expected $\tO(\Delta\ell\sigma\sqrt{d}\epsilon^{-3})$ queries.
\end{theorem}

The proof of \thm{Q-SGD} is based on the following result from~\cite{ghadimi2013stochastic}.

\begin{theorem}[{\cite[Theorem 1]{ghadimi2013stochastic}}]\label{thm:SGD-convergence}
Consider the bounded variance setting in \prob{nonconvex}, the output $\x_{\mathrm{out}}$ of \algo{Q-SGD} is an expected $\frac{2\Delta\ell}{\mathcal{T}}+\hat{\sigma}^2$-critical point.
\end{theorem}

\begin{proof}[Proof of \thm{Q-SGD}]
By \thm{SGD-convergence} and the associated parameter setting, the output $\x_{\mathrm{out}}$ of \algo{Q-SGD} satisfies
\begin{align*}
\E\|\nabla f(\x_{\mathrm{out}})\|^2\leq \frac{2\Delta\ell}{\mathcal{T}}+\hat{\sigma}^2\leq\epsilon^2,
\end{align*}
implying that \algo{Q-SGD} solves \prob{nonconvex} with probability at least $2/3$ under the bounded variance setting, and the remaining thing would be to analyze the number of queries it makes to $O_\g$. By \thm{unbiased-estimator}, at each iteration one can run \algo{unbiased-Q} to obtain an unbiased estimate $\tilde{\g}_t$ of $\nabla f(\x_t)$ with variance at most $\hat{\sigma}^2$ using
\begin{align*}
\tO\left(\frac{\sigma\sqrt{d}}{\hat{\sigma}}\right)=\tO\left(\frac{\sigma\sqrt{d}}{\epsilon}\right)
\end{align*}
queries to $O_\g$. Then, the total number of queries to $O_\g$ is at most
\begin{align*}
\mathcal{T}\cdot\tO\left(\frac{\sigma\sqrt{d}}{\epsilon}\right)=\tO\left(\frac{\sigma\Delta\ell\sqrt{d}}{\epsilon^3}\right).
\end{align*}
\end{proof}

\begin{algorithm2e}
	\caption{Q-SPIDER}
	\label{algo:Q-SPIDER}
	\LinesNumbered
	\DontPrintSemicolon
    \KwInput{Function $f\colon\R^d\to\R$, precision $\epsilon$, variance $\sigma$, smoothness $\ell$}
    \KwParameter{$q=\frac{20\sigma}{\epsilon}$, $\hat{\sigma}_1=\frac{\epsilon}{40}$, $\hat{\sigma}_2=\frac{\epsilon}{40}\sqrt{\frac{\epsilon}{10\sigma}}$, total iteration budget $\mathcal{T}=\frac{1600\ell\Delta}{\sigma^2}$}
    \KwOutput{An $\epsilon$-critical point of $f$}
    Set $\x_0\leftarrow \0$\;
    \For{$t=0,1,2,\ldots,\mathcal{T}$}{
        \If{$\mod(t,q)=0$}{
            Call \algo{unbiased-Q} to obtain an unbiased estimate $\tilde{\g}_t$ of $\nabla f(\x_t)$ with variance $\hat{\sigma}_1^2$
            \;
            $\vect{v}_{t}\leftarrow\tilde{\g}_t$\label{lin:outer-QVR}\;
        }
        \Else{
            Call \algo{unbiased-Q} to obtain 
            an unbiased estimate $\tilde{\g}_t$ of $\nabla f(\x_t)-\nabla f(\x_{t-1})$ with variance $\hat{\sigma}_2^2$\;\label{lin:inner-QVR}
            $\vect{v}_t\leftarrow\tilde{\g}_t+\vect{v}_{t-1}$\;
        }
        \lIf{$\|\vect{v}_t\|\leq 2\epsilon$}{
            \Return $\x_t$;
        }
        \lElse{
            $\x_{t+1}\leftarrow\x_t-\frac{\epsilon}{\ell}\cdot\frac{\vect{v}_t}{\|\vect{v}_t\|}$
        }
    }
    \Return $\x_{\mathcal{T}}$\;
\end{algorithm2e}

\subsection{Algorithm for the mean-squared smoothness setting}\label{sec:mean-squared-smoothness-algorithm}

To solve \prob{nonconvex} in the mean-squared smoothness setting, we leverage the  SPIDER algorithm introduced in~\cite{fang2018spider}, which is a variance reduction technique that allows us to estimate the gradient of a function with lower cost by utilizing the smoothness structure and reuse the stochastic gradient samples at nearby points. Our algorithm is a specialization of the SPIDER algorithm, where we replace the classical variance reduction step by quantum variance reduction (\algo{unbiased-Q}). The query complexity of our quantum algorithm is given in the following theorem.

\begin{theorem}[Formal version of \thm{Q-SPIDER-informal}, mean-squared smoothness setting]\label{thm:Q-SPIDER}
For any $0\leq\epsilon\leq\sigma$, \algo{Q-SPIDER} solves \prob{nonconvex} in the mean-squared smoothness setting using an expected $\tO\left(\ell\Delta\sqrt{d\sigma}\epsilon^{-2.5}\right)$ number of queries. 
\end{theorem}

The proof of \thm{Q-SPIDER} is based on the following result from~\cite{fang2018spider}.

\begin{theorem}[{\cite[Theorem 1]{fang2018spider}}]\label{thm:SPIDER-convergence}
Consider the mean-squared smoothness setting in \prob{nonconvex}, the output $\x_{\mathrm{out}}$ of \algo{Q-SPIDER} is an expected $\epsilon$-critical point.
\end{theorem}

\begin{proof}[Proof of \thm{Q-SPIDER}]
\thm{SPIDER-convergence} shows that \algo{Q-SPIDER} can solve \prob{nonconvex} under the mean-squared smoothness setting and the remaining thing would be to analyze the number of queries it makes to $O_\g^S$. 

For each iteration $t$ with $\mathrm{mod}(t,q)=0$ such that \lin{outer-QVR} is executed, note that we can prepare the following quantum oracle
\begin{align*}
\hat{O}_\g\ket{\x}\otimes\ket{0}\to\ket{\x}\otimes\int_{\omega}\sqrt{p(\omega)\d\omega}\ket{\g(\x,\omega)}\otimes\ket{\mathrm{garbage}(\x,\omega)},
\end{align*}
where $p(\omega)$ is the probability distribution of the random seed $\omega$, by first applying $O_\g^{S}$ to the state
\begin{align*}
\ket{\x}\otimes\left(\int_{\omega}\sqrt{p(\omega)\d\omega}\ket{\omega}\right)\otimes\ket{0}
\end{align*}
and then uncompute the second register. This procedure uses one query to the $\sigma$-SQ-QSGO $O_\g^S$ (\defn{OgS}). By \thm{unbiased-estimator}, one can run \algo{unbiased-Q} to obtain an unbiased estimate $\tilde{\g}_t$ of $\nabla f(\x_t)$ with variance at most $\hat{\sigma}_1^2$ using
\begin{align*}
\tO\left(\frac{\sigma\sqrt{d}}{\hat{\sigma}_1}\right)=\tO\left(\frac{\sigma\sqrt{d}}{\epsilon}\right)
\end{align*}
queries to $\hat{O}_\g$ and thus the same number of queries to $O_{\g}^{S}$.

Similarly, for each iteration $t$ with $\mathrm{mod}(t,q)=0$ such that \lin{inner-QVR} is executed, note that one can prepare the following quantum oracle
\begin{align*}
\hat{O}_{\g}^t\ket{\x_t}\otimes\ket{\x_{t-1}}\otimes\ket{0}&\to\ket{\x_t}\otimes\ket{\x_{t-1}}\\
&\qquad\otimes\int_\omega\sqrt{p(\omega)\d\omega}\ket{\g(\x_t,\omega)-\g(\x_{t-1},\omega)}\otimes\ket{\mathrm{garbage}(\x_t,\x_{t-1},\omega)}
\end{align*}
by first applying $O_\g^{S}$ twice to obtain the state
\begin{align*}
&\ket{\x}\otimes\ket{\x_{t-1}}\otimes\int_\omega\sqrt{p(\omega)\d\omega}\ket{\omega}\otimes\ket{\g(\x,\omega)}\otimes\ket{\g(\x_{t-1},\omega)}\\
&\qquad\qquad\qquad\otimes\ket{\g(\x,\omega)-\g(\x_{t-1},\omega)}\otimes\ket{\mathrm{garbage}(\x_t,\x_{t-1},\omega)}
\end{align*}
and then uncompute the fourth and the fifth register. Observe that
\begin{align*}
\underset{\omega}{\E}\|\g(\x_t,\omega)-\g(\x_{t-1},\omega)\|^2\leq\ell^2\|\x_t-\x_{t-1}\|\leq\epsilon^2,
\end{align*}
then by \thm{unbiased-estimator}, one can run \algo{unbiased-Q} to obtain an unbiased estimate $\tilde{\g}_t$ of $\nabla f(\x_t)-\nabla f(\x_{t-1})$ with variance at most $\hat{\sigma}_1^2$ using
\begin{align*}
\tO\left(\frac{\epsilon\sqrt{d}}{\hat{\sigma}_2}\right)=\tO\left(\sqrt{\frac{d\sigma}{\epsilon}}\right)
\end{align*}
queries to $\hat{O}_\g^{t}$ and thus twice the number of queries to $O_{\g}^S$. Then, the total number of queries to $O_{\g}^S$ can be expressed as
\begin{align*}
\frac{\mathcal{T}}{q}\cdot\tO\left(\frac{\sigma\sqrt{d}}{\epsilon}\right)+\mathcal{T}\cdot\tO\left(\sqrt{\frac{d\sigma}{\epsilon}}\right)=\tO\left(\frac{\ell\Delta\sqrt{d}}{\epsilon^2}\left(1+\sqrt{\frac{\sigma}{\epsilon}}\right)\right).
\end{align*}
\end{proof}


\input{sec-lower}

\input{sec-conclusion}

\bibliographystyle{myhamsplain}
\bibliography{quantum-sco}

\end{document}

%% file: macros.tex
\newcommand{\eqn}[1]{(\ref{eqn:#1})}

\newcommand{\thm}[1]{\hyperref[thm:#1]{Theorem~\ref*{thm:#1}}}
\newcommand{\cor}[1]{\hyperref[cor:#1]{Corollary~\ref*{cor:#1}}}
\newcommand{\defn}[1]{\hyperref[defn:#1]{Definition~\ref*{defn:#1}}}
\newcommand{\lem}[1]{\hyperref[lem:#1]{Lemma~\ref*{lem:#1}}}
\newcommand{\prop}[1]{\hyperref[prop:#1]{Proposition~\ref*{prop:#1}}}
\newcommand{\assum}[1]{\hyperref[assum:#1]{Assumption~\ref*{assum:#1}}}
\newcommand{\fig}[1]{\hyperref[fig:#1]{Figure~\ref*{fig:#1}}}
\newcommand{\tab}[1]{\hyperref[tab:#1]{Table~\ref*{tab:#1}}}
\newcommand{\algo}[1]{\hyperref[algo:#1]{Algorithm~\ref*{algo:#1}}}
\renewcommand{\sec}[1]{\hyperref[sec:#1]{Section~\ref*{sec:#1}}}
\newcommand{\append}[1]{\hyperref[append:#1]{Appendix~\ref*{append:#1}}}
\newcommand{\fac}[1]{\hyperref[fac:#1]{Fact~\ref*{fac:#1}}}
\newcommand{\lin}[1]{\hyperref[lin:#1]{Line~\ref*{lin:#1}}}
\newcommand{\prob}[1]{\hyperref[prob:#1]{Problem~\ref*{prob:#1}}}

\def\>{\rangle}
\def\<{\langle}

\newcommand{\vect}[1]{\ensuremath{\mathbf{#1}}}
\newcommand{\x}{\ensuremath{\mathbf{x}}}
\newcommand{\y}{\ensuremath{\mathbf{y}}}
\newcommand{\z}{\ensuremath{\mathbf{z}}}
\newcommand{\g}{\ensuremath{\mathbf{g}}}
\renewcommand{\v}{\ensuremath{\mathbf{v}}}
\newcommand{\xag}{\ensuremath{\mathbf{x}^{ag}}}
\newcommand{\xmd}{\ensuremath{\mathbf{x}^{md}}}

\newcommand{\N}{\mathbb{N}}
\newcommand{\B}{\mathbb{B}}

\newcommand{\R}{\mathbb{R}}

\newcommand{\E}{\mathbb{E}}

\newcommand{\tO}{\widetilde{\mO}}
\newcommand{\mO}{\mathcal{O}}
\newcommand{\mC}{\mathcal{C}}

\DeclareMathOperator{\Var}{Var}

\DeclareMathOperator{\out}{out}
\DeclareMathOperator{\sls}{\mathtt{StochasticLineSearch}}

\newcommand{\search}{\mathsf{Search}}
\newcommand{\parity}{\mathsf{Parity}}
\renewcommand{\d}{\mathrm{d}}
\renewcommand{\x}{\vect{x}}

\newcommand{\e}{\vect{e}}
\newcommand{\0}{\mathbf{0}}

\def\Tr{\operatorname{Tr}}\def\:{\hbox{\bf:}}

\SetKwInput{KwInput}{Input} 
\SetKwInput{KwParameter}{Parameters}                %
\SetKwInput{KwOutput}{Output}              %
\SetKwInput{KwReturn}{Return}

\let\oldnl\nl
\newcommand{\nonl}{\renewcommand{\nl}{\let\nl\oldnl}}

\newcommand{\defeq}{:=}

%% file: sec-intro.tex
\section{Introduction}
\label{sec:intro}

Stochastic optimization is central to modern machine learning. Stochastic gradient descent (SGD), and its many variants, are used broadly for solving challenges in data science and learning theory. In theory, SGD and stochastic optimization, have been the subject of decades of extensive study. \cite{cesa2004generalization,nemirovski2009robust,duchi2018introductory} established that SGD achieves optimal rates for minimizing Lipschitz convex functions\footnote{We assume all objective functions are differentiable. Similar to related work, see e.g., \cite{bubeck2019complexity}, our results can be generalized to non-differentiable settings since convex functions are almost everywhere differentiable and our algorithms and corresponding convergence analysis are robust to polynomially small numerical errors.
} (even in one dimension) and stochastic optimization methods have been established for a range of 
problems~\cite{polyak1987introduction,nesterov2003introductory,shalev2007online,shalev2014understanding}. More recently, the complexity of stochastic gradient methods for smooth non-convex optimization, e.g., critical point computation, were established~\cite{allen2018neon2,fang2018spider,fang2019sharp,jin2017escape,jin2019stochastic}.

Given the foundational nature of stochastic optimization and the potential promise and increased study of quantum algorithms, it is natural to ask whether quantum computation could enable improved rates for solving these problems. There has been work studying whether access to quantum counterparts of classic optimization oracle can yield faster rates for semidefinite programs~\cite{brandao2017quantum,brandao2017SDP,gong2022robustness,van2019improvements,van2020quantum,kerenidis2018interior}, convex optimization~\cite{chakrabarti2020optimization,vanApeldoorn2020optimization,chakrabarti2023quantum}, and non-convex optimization~\cite{childs2022quantum,liu2023quantum,zhang2021quantum}. Notably, \cite{jordan2005fast} showed that with just access to the quantum analog of an \emph{zeroth-order oracle}, i.e., an oracle that when queried at a point outputs the value of the function at that point, it is possible to simulate access to a classic gradient oracle with a single query. This tool immediately yields improved rates for, e.g., convex function minimization, with a zeroth-order oracle.

Unfortunately, despite this progress and the established power of quantum evaluation oracles, obtaining further improvements has been challenging. A line of work established a variety of striking lower bounds ruling out quantum speedups for fundamental optimization problems~\cite{garg2020no,garg2021near,zhang2022quantum}. For example, \cite{garg2020no} showed that when given access to the quantum analog of a \textit{first-order oracle}, i.e., an oracle that when queried at a point outputs the value of the function as well as the gradient at that point, quantum algorithms have no improved rates for non-smooth convex optimization over GD and SGD when the dimension is large. \cite{zhang2022quantum} extended this result to the non-convex setting and showed that given access to the quantum analog of a stochastic gradient oracle, quantum algorithms have no improved rates for finding critical points over SGD when the dimension is large.

In spite of these negative results, we nevertheless ask, \emph{for stochastic optimization are quantum-speedups obtainable?} Our main result is an answer to this question in the affirmative for dimension-dependent algorithms. We provide two different quantum algorithms for stochastic convex optimization (SCO) which provably outperform optimal classic algorithms. Furthermore, we provide a quantum algorithm for computing the critical point of a smooth non-convex function which improves upon the state-of-the-art. We obtain these results through a new general quantum-variance reduction technique built upon the quantum multivariate mean estimation result of Cornelissen et al.~\cite{cornelissen2022near} and the multilevel Monte Carlo (MLMC) technique~\cite{giles2015multilevel,blanchet2015unbiased,asi2021stochastic}. We complement these results with lower bounds showing that one of them is asymptotically optimal in low-dimensional settings. 

\paragraph{General notation.}
We use $\|\cdot\|$ to denote the Euclidean norm and let $\B_R(\x) \defeq \{\y\in\R^d\colon\|\y-\x\|\leq R\}$ and $[T] \defeq \{1,\ldots,T\}$. We use bold letters, e.g., $\x,\y$, to denote vectors and capital letters, e.g., $A,B$, to denote matrices. For a $d$-dimensional random variable $X$, we refer to the trace of the covariance matrix of $X$ as its variance, denoted by $\mathrm{Var}[X]$. For $f\colon\R^d\to\R$, we let $f^* \defeq \inf_{\x}f(\x)$ and call $\x \in \R^d$ \emph{$\epsilon$-(sub)optimal} if $f(\x) \leq f^* + \epsilon$ and \emph{$\epsilon$-critical} if $\norm{\nabla f(\x)} \leq \epsilon$. Moreover, we call a random point $\x\in\R^d$ \emph{expected $\epsilon$-(sub)optimal} if $\E f(\x) \leq f^* + \epsilon$ and \emph{expected $\epsilon$-critical} if $\E\norm{\nabla f(\x)} \leq \epsilon$. $f\colon\R^d\to\R$ is \emph{$L$-Lipschitz} if
$f(\x)-f(\y)\leq L\|\x-\y\|$ for all $\x,\y\in\R^d$ and \emph{$\ell$-smooth} if 
$\|\nabla f(\x)-\nabla f(\y)\|\leq \ell\|\x-\y\|$ for all $\x,\y\in\R^d$. For two matrices $A,B\in\R^{d\times d}$, $A \preceq B$ denotes that $\x^T A \x \leq \x^T B \x$ for all $\x\in\R^d$. We use $\tO$ to denote the big-$\mO$ notation omitting poly-logarithmic factors in $\epsilon,\epsilon',d,\sigma,\hat{\sigma},R,$ and $L$. When applicable, we use $\ket{\mathrm{garbage}(\cdot)}$ to denote possible garbage states\footnote{The garbage state is a quantum analogue of classical garbage information that arises when preparing the classical stochastic gradient oracle which cannot be erased or uncomputed in general. In this work, we consider a general model where we make no assumption on the garbage state. See e.g.,~\cite{gilyen2020distributional} for a similar discussion of this standard use of garbage quantum states.} that arise during the implementation of a quantum oracle (e.g., in Definitions \eqn{OX-defn} and \eqn{quantum-SG-oracle}).

\subsection{Quantum stochastic optimization oracles}
\label{sec:intro:quantum}

Here we formally define quantum stochastic optimization oracles that we study in this work.

\paragraph{Qubit notation.}
We use $\ket{\cdot}$ to represent input or output registers made of qubits that could be in \textit{superpositions}. In particular, given $m$ points $\x_1,\ldots,\x_m\in\R^d$ and a coefficient vector $\vect{c}\in\mathbb{C}^m$ with $\sum_{i \in [m]} |c_i|^2=1$, the quantum register could be in the quantum state $\ket{\psi}=\sum_{i \in [m]} c_i\ket{\x_i}$, which is superposition over all these $m$ points at the same time. If we measure this state, we will get each $\x_i$ with probability $|c_i|^2$. Furthermore, to model a classical probability distribution $p$ over $\R^d$ quantumly, we can prepare the quantum state $\int_{\x\in\R^d}\sqrt{p(\x)\d\x}\ket{\x}$. If we measure this state, the measurement outcome would follow the probability density function $p$.

\paragraph{Quantum random variable access.}
We say that we have \emph{quantum access to a $d$-dimensional random variable $X$} if we can query the following \textit{quantum sampling oracle} of $X$ that returns a quantum superposition over the probability distribution of $X$ defined as follows.\footnote{Throughout this paper, whenever we have access to a quantum oracle $O$, we assume that it is a unitary operation and that we also have access to its corresponding inverse operation, denoted as $O^{-1}$, that satisfies $O^{-1}O=OO^{-1}=I$. These are a standard assumption, either explicitly or implicitly, in prior work on quantum algorithms, see e.g.,~\cite{brassard2002quantum,hamoudi2021quantum,cornelissen2022near}.}
\begin{definition}[Quantum sampling oracle]\label{defn:OX}
For a $d$-dimensional random variable $X$, its quantum sampling oracle $O_X$ is defined as
\begin{align}\label{eqn:OX-defn}
O_X\ket{0}\rightarrow\int_{\x\in\R^d}\sqrt{p_X(\x)\d\x}\ket{\x}\otimes\ket{\mathrm{garbage}(\x)},
\end{align}
where $p_X(\cdot)$ represents the probability density function of $X$.
\end{definition}

The garbage state in \defn{OX} is a quantum analogue of classical garbage information that arises when preparing the classical sampling oracle of $X$. When implementing the quantum sampling oracle in quantum superpositions however, this garbage information will appear in a quantum state and cannot be erased or uncomputed in general. In this work, we consider a general model where we make no assumption on the garbage state. See e.g.,~\cite{gilyen2020distributional} for a similar discussion of this standard use of garbage quantum states.

Observe that if we directly measure the output of $O_X$, it will collapse to a classical sampling access to $X$ that returns random vectors with respect to the probability distribution $p_X$.

\paragraph{Quantum stochastic gradient oracle.}
When considering the problem of optimizing a function $f\colon\R^d\to\R$, we are often given access to a stochastic gradient oracle that returns a random vector from the probability distribution of the stochastic gradient.

\begin{definition}[Stochastic gradient oracle (SGO)]\label{defn:Cg}
For $f\colon\R^d\to\R$, its \emph{stochastic gradient oracle (SGO)}
$\mC_\g$ is defined as a random function that when queried at $\x$, samples a vector $\g(\x)$ from a probability distribution $p_{f,\x}(\cdot)$ over $\R^d$ that satisfies 
\begin{align*}
\underset{\g(\x)\sim p_{f,\x}}{\E}\g(\x)=\nabla f(\x),\quad\forall\x\in\R^d.
\end{align*}
We say the oracle is \emph{$L$-bounded} if
\begin{align*}
\underset{\g(\x)\sim p_{f,\x}}{\E}\|\g(\x)\|^2\leq L^2,\quad\forall\x\in\R^d,
\end{align*}
and we say the oracle has \emph{variance $\sigma^2$} if
\begin{align*}
\underset{\g(\x)\sim p_{f,\x}}{\E}\|\g(\x)-\nabla f(\x)\|^2\leq \sigma^2,\quad\forall\x\in\R^d.
\end{align*}
\end{definition}

In this paper, we further assume quantum access to a stochastic gradient oracle, or access to a \textit{quantum stochastic gradient oracle} for brevity, that upon query returns a quantum superposition over the probability distribution $p_{f,\x}(\v)$. 

\begin{definition}[Quantum stochastic gradient oracle (QSGO)]\label{defn:Og}
For $f\colon\R^d\to\R$, its quantum stochastic gradient oracle (QSGO) is defined as
\begin{align}\label{eqn:quantum-SG-oracle}
O_\g\ket{\x}\otimes\ket{0}\to\ket{\x}\otimes\int_{\v\in\R^d}\sqrt{p_{f,\x}(\v)\d \v}\ket{\v}\otimes\ket{\mathrm{garbage}(\v)},
\end{align}
where $p_{f,\x}(\cdot)$ is as defined in \defn{Cg}.
\end{definition}
Observe that if we directly measure the output of $O_\g$, it will collapse to a classical stochastic gradient oracle that randomly returns a stochastic gradient at $\x$. 

Another standard assumption in previous works~\cite{agarwal2017finding,fang2018spider,fang2019sharp,jin2019stochastic} that the stochastic gradient can be queried \textit{simultaneously}, which means that the algorithm can choose the random seed $\omega$ that is being queried. We assume that in such setting there is an explicit probability distribution $\omega$ such that $\underset{\omega}{\E}[\g(\x,\omega)]=\nabla f(\x)$. Similarly, we define the quantum access to stochastic gradients allowing simultaneous queries, or access to a \textit{quantum stochastic gradient oracle with simultaneous queries} for brevity, that upon query returns $\g(\x,\omega)$ in a quantum state.

\begin{definition}[$\sigma$-SQ-QSGO]\label{defn:OgS}
For $f\colon\R^d\to\R$ with its stochastic gradient $\g(\x,\omega)$ indexed by random seed $\omega$ that satisfies
\begin{align*}
\underset{\omega}{\E}\,\g(\x,\omega)=\nabla f(\x)\quad\text{,}\quad\underset{\omega}{\E}\|\g(\x,\omega)-\nabla f(\x)\|^2\leq\sigma^2,\quad\forall\x\in\R^d,\quad\text{and}
\end{align*}
\begin{align}\label{eqn:mss}
\underset{\omega}{\E}\|\g(\x,\omega)-\g(\y,\omega)\|^2\leq\ell^2\|\x-\y\|^2 ,\quad\forall\x,\y\in\R^d\,
\end{align}
its simultaneously queriable, \emph{$\sigma$-mean-squared smooth quantum stochastic gradient oracle ($\sigma$-SQ-QSGO)} is defined as
\begin{align}\label{eqn:OgS}
O_{\g}^{S}\ket{\x}\otimes\ket{\omega}\otimes\ket{0}\to\ket{\x}\otimes\ket{\omega}\otimes\ket{\g(\x,\omega)}\otimes\ket{\mathrm{garbage}(\x,\omega)}.
\end{align}
\end{definition}

\subsection{Results}
\label{sec:intro:results}
Here we present our main results on new quantum algorithms for stochastic optimization. Our results and the prior state-of-the-art are summarized in \tab{SCO-algorithms}. Further, we discuss new quantum lower bounds that we establish for quantum variance reduction and stochastic convex optimization. 
Our algorithmic results leverage a common technique for quantum variance reduction introduced in the next \sec{QVR}. This technique uses a combination of the quantum multivariate mean estimation result of Cornelissen et al.~\cite{cornelissen2022near} and multilevel Monte Carlo (MLMC)~\cite{giles2015multilevel,blanchet2015unbiased,asi2021stochastic}.

\begin{table}[ht]
\centering
\resizebox{1.0\columnwidth}{!}{
\begin{tabular}{cccc}
\hline
Setting & Queries & Output & Method \\ \hline \hline
Convex & $\epsilon^{-2}$ & $\epsilon$-optimal & SGD \\ \hline
Convex & $d^{5/8} \epsilon^{-3/2}$ & $\epsilon$-optimal & \textbf{Our Result} (Q-AC-SA, \algo{Q-AC-SA}) \\ \hline
Convex & $d^{3/2} \epsilon^{-1}$ & $\epsilon$-optimal & 
\makecell{\textbf{Our Result}
(Q-SCP, \sec{QSCPM})} \\ \hline
\makecell{Non-convex\\ (Bounded variance)} & $\epsilon^{-4}$ & $\epsilon$-critical & Randomized SGD~\cite{ghadimi2013stochastic} \\ \hline
\makecell{Non-convex\\ (Bounded variance)} & $d^{1/2}\epsilon^{-3}$ & $\epsilon$-critical &
\textbf{Our Result}
(Q-SGD, \algo{Q-SGD}) \\ \hline
\makecell{Non-convex\\ (Mean-squared smoothness)} & $\epsilon^{-3}$ & $\epsilon$-critical & SPIDER~\cite{fang2018spider} \\ \hline
\makecell{Non-convex\\ (Mean-squared smoothness)} & $d^{1/2} \epsilon^{-5/2}$ & $\epsilon$-critical &
\textbf{Our Result}
(Q-SPIDER, \algo{Q-SPIDER}) \\ \hline
\end{tabular}
}
\vspace{1mm}
\caption{Comparison between our quantum algorithms and state-of-the-art classical algorithms for \prob{SCO} (the convex setting) and \prob{nonconvex} (the non-convex setting) where, for simplicity, we assume $R=L=1$ for \prob{SCO} and $\ell=\sigma=\delta=1$ for \prob{nonconvex}. 
The $\tO$ symbol was omitted in the ``Queries'' column. 
}
\label{tab:SCO-algorithms}
\end{table}

\paragraph{Stochastic convex optimization.}In this work we consider the quantum analog of the standard stochastic convex optimization (SCO) problem defined as follows.

\begin{problem}[Quantum stochastic convex optimization (QSCO)]\label{prob:SCO}
In the quantum stochastic convex optimization (QSCO) 
problem we are given query access to an 
$L$-bounded QSGO $O_\g$ (see \defn{Og}) for a convex function $f\colon\R^d\to\R$ whose minimum is achieved at $\x^*$ with $\|\x^*\|\leq R$ and must output an expected $\epsilon$-optimal point.
\end{problem}

Classically, it is known that simple stochastic gradient descent, e.g. $x_{t + 1} = x_t - \eta g_t$, can solve QSCO with $\mO(\epsilon^{-2})$ queries. Further, this bound is known to be optimal in the worst case~\cite{nemirovskij1983problem,woodworth2017lower}. 

Nevertheless, we develop two quantum algorithms for \prob{SCO} in \sec{Q-AC-SA} and \sec{QSCPM}, respectively. The query complexities of these algorithms are summarized in the following \thm{SCO-informal}. In comparison to the optimal classical query complexity of $\mO(\epsilon^{-2})$ queries, our algorithms obtain an improved dependence in terms of $\epsilon$ at the cost of a worse dependence on the dimension, $d$; \thm{SCO-informal} shows that a quadratic speedup is achievable when $d$ is constant. 

\begin{theorem}[Informal version of \thm{Q-AC-SA} and \cor{QSCPM}]\label{thm:SCO-informal}
\prob{SCO} can be solved using an expected $\tO(\min\{ d^{5/8}(LR/\epsilon)^{3/2}, d^{3/2}LR/\epsilon\})$ 
queries.
\end{theorem}

We complement \thm{SCO-informal} with the following  lower bound on the query complexity for QSCO.

\begin{theorem}[Informal version of \thm{SCO-lower}]\label{thm:SCO-lower-informal}
For any $\epsilon\leq\mO(d^{-1/2})$, any quantum algorithm that solve \prob{SCO} with probability at least $2/3$ makes at least $\Omega(d^{1/2}\epsilon^{-1})$ queries in the worst case. For any $\Omega(d^{-1/2})\leq\epsilon\leq 1$, any quantum algorithm that solves \prob{SCO} with success probability at least $2/3$ must make at least $\Omega(\epsilon^{-2})$ queries in the worst case.
\end{theorem}

\thm{SCO-lower-informal} shows that the $\tO(d^{3/2} LR / \epsilon)$ rate that we obtain is asymptotically optimal for $d = \tO(1)$. A key open problem is whether the dimension dependence in either our upper bounds (\thm{Q-AC-SA} and \cor{QSCPM}) or our lower bound (\thm{SCO-lower}) can be improved.

\paragraph{Stochastic critical point computation.}

We also develop quantum algorithms for finding critical points, i.e., points with small gradients, of (possibly) non-convex functions.

\begin{problem}[Quantum stochastic critical point computation (QSCP)]\label{prob:nonconvex}
In the quantum stochastic critical point computation (QSCP) problem, for an $\ell$-smooth (posssibly) non-convex $f\colon\R^d\to\R$ satisfying $f(\0)-\inf_{\x}f(\x)\leq\Delta$ we are given query access to one of the following two oracles
\begin{enumerate}
\item (Bounded variance setting). A QSGO $O_{\g}$ with variance $\sigma^2$ (see \defn{Og}), or
\item (Mean-squared smoothness setting). A $\sigma$-SQ-QSGO $O_{\g}^S$ (see \defn{OgS}), 
\end{enumerate}
and must output an expected $\epsilon$-critical point.
\end{problem}

Leveraging \cite{ghadimi2013stochastic} and~\cite{fang2018spider} we develop two quantum algorithms that solve \prob{nonconvex} in the bounded variance setting and the mean-squared smoothness setting, and obtain the following result.

\begin{theorem}[Informal version of \thm{Q-SGD} and \thm{Q-SPIDER}]\label{thm:Q-SPIDER-informal}
In the bounded variance setting, \prob{nonconvex} can be solved using an expected $\tO\left(\Delta\ell\sigma d^{1/2}\epsilon^{-3}\right)$ queries. In the mean-squared smoothness setting, \prob{nonconvex} can be solved using an expected $\tO\left(\ell\Delta(d\sigma)^{1/2} \epsilon^{-5/2}\right)$ queries.
\end{theorem}
In the bounded variance setting, \prob{nonconvex} can be solved using $\mO(\epsilon^{-4})$ queries to a classical SGO with variance $\sigma^2$  (\defn{Cg})~\cite{ghadimi2013stochastic}, which is known to be optimal~\cite{arjevani2022lower}. In comparison, our algorithm improves in terms of $\epsilon$ and achieves a quantum speedup when $d\leq\mO(\epsilon^{-2})$. In the mean-squared smoothness setting, \prob{nonconvex} can be solved using $\mO(\epsilon^{-3})$ queries to a classical stochastic gradient oracle with variance $\sigma^2$ that satisfies if it satisfies \eqn{mss} (\defn{Cg})~\cite{fang2018spider}, which is known to be optimal~\cite{arjevani2022lower}. In comparison, our algorithm improves in terms of $\epsilon$ and achieves a quantum speedup when $d\leq\mO(\epsilon^{-1})$.

\paragraph{Quantum zeroth-order oracles.}
Throughout the paper, we focus on stochastic gradient oracles in correspondence with classical work on stochastic optimization. However, it is worth noting that in certain cases our results extend gracefully to quantum stochastic zeroth order oracles. For example, when the objective function exhibits a finite-sum structure and we have access each component function individually through a quantum zeroth-order oracle, we can achieve an SQ-QSGO (\defn{OgS}) with just a single query, utilizing quantum gradient estimation~\cite{jordan2005fast}. However, in other cases, the correspondence is less clear. For instance, if we are given a quantum stochastic zeroth-order oracle where the function value is obfuscated by some external noise, quantum gradient estimation~\cite{jordan2005fast} is not directly applicable. Further study could be an interesting direction for future work.

\paragraph{Practicality.} Regarding the utility of our algorithm in practical situations, note that our quantum oracles in \defn{Og} and \defn{OgS} are defined as direct, natural generalizations of the corresponding classical oracles. Considering such quantum generalizations of classical oracles is standard in the literature, see e.g.,~\cite{chakrabarti2020optimization,zhang2021escape}. There are standard techniques for implementing such quantum analogs of classical oracles (in theory for now given the current state-of-the-art in implementing quantum algorithms in practice). In particular, if there is a classical circuit for the classical oracle, there is a standard technique to obtain a quantum circuit of the same size which implement the corresponding quantum oracle and its inverse. Hence, our quantum algorithms have the potential to surpass blackbox classical algorithms in low dimensional settings where the oracle is given as an explicit circuit. 

\paragraph{Dimension dependence.} Regarding potential concerns regarding the dependence of our quantum algorithms on the problem dimension, below we provide several supplementary points of context. 
\begin{itemize}

\item As discussed, in the classical setting, prior research~\cite{duchi2018introductory} demonstrated that when optimizing an 1-Lipschitz convex function, SGD has a query complexity of $\mO(\epsilon^{-2})$ which is optimal, even in the one-dimensional case. In the quantum setting, we showed that, theoretically, quantum speedups which offer a different tradeoff between $\epsilon$ and $d$ dependencies are possible. Additionally, from our lower bound presented in \thm{SCO-lower-informal} we know that some dimension dependence is inherent in obtaining an improvement.
\item Classically, the complexity of dimension dependent optimization methods is well studied. In particular, there are parallel and private stochastic convex settings where the dimension dependencies are discussed, see e.g.,~\cite{bubeck2019complexity,carmon2023resqueing}, and there exists works on critical point computation in low dimension settings, see e.g.,~\cite{chewi2023complexity}.
\item Even for high-dimensional problems, our algorithms can potentially be used as a subroutines for low-dimensional subproblems. For instance, in this paper we apply these method to the approximately best point problem (\prob{finite-optimal-finding}), wherein we utilize our algorithm repeatedly within a one-dimensional setting.
\end{itemize}

\subsection{Additional related work}
\paragraph{Stochastic convex optimization.}
Stochastic convex optimization is a broad, well-studied area of research. Beyond the works mentioned in the introduction, for additional references and discussion of research in the area, see \cite{boyd2008stochastic,hastie2009elements,nemirovski2009robust} for detailed overviews of stochastic convex optimization methods and \cite{cesa2004generalization} for a presentation of convex optimization in the online learning setting. Additionally, see \cite{defazio2014saga,johnson2013accelerating} for a discussion of the closely related problem of finite-sum optimization.

\paragraph{Non-convex optimization.}
Non-convex optimization is a rapidly advancing research area in optimization theory. These advances are motivated in part by the fact that the landscapes of various modern machine learning problems are typically non-convex, including deep neural networks, principal component analysis, tensor decomposition, etc. In general, finding a global optimum of a non-convex function is NP-hard~\cite{murty1985some,nemirovskij1983problem}. Hence, many theoretical works instead focus on finding local minima rather than a global one, given existing empirical and theoretical evidence that local minima can be as good as global minima in certain machine learning problems~\cite{bhojanapalli2016global,ge2015escaping,ge2016matrix,ge2018learning,ge2017optimization,hardt2018gradient}. The first step of finding local minima would be to find stationary points, which has been discussed in many previous works~\cite{agarwal2017finding,birgin2017worst,carmon2018accelerated,carmon2017convex,nesterov2003introductory,nesterov2006cubic}. An important line of work focuses on first-order, (stochastic) gradient-based algorithms for finding critical points and local minima since higher-order derivatives are often not accessible in practical scenarios~\cite{allen2018neon2,carmon2018accelerated,fang2018spider,fang2019sharp,ge2015escaping,jin2017escape,jin2019stochastic,liu2018adaptive,xu2017neon,zhang2021escape}. There are also results on non-convex optimization that study different settings~\cite{ge2019stabilized,li2019ssrgd,yu2018third,zhou2020stochastic,ge2017no,zhang2018primal,zhou2018stochastic}.

\paragraph{Quantum Monte Carlo methods.}
Since~\cite{montanaro2015quantum}, quantum Monte Carlo methods have been broadly investigated. In particular, \cite{li2018quantum} discussed quantum Monte-Carlo methods for entropy estimation, and~\cite{hamoudi2019quantum} initiated the study of quantum mean estimation problem with multiplicative error, which is followed by~\cite{hamoudi2021quantum,cornelissen2022near,cornelissen2023sublinear,kothari2023mean}. Multilevel Monte Carlo methods are also widely used in quantum algorithms, see e.g.,~\cite{an2021quantum,li2022enabling}. Specifically, the particular multilevel Monte Carlo strategies we use for quantum variance reduction have roots in the classic computation work of~\cite{blanchet2015unbiased,asi2021stochastic}.

\subsection{Paper organization} 
In the next \sec{QVR} we propose and develop a new algorithm for a new problem \textit{quantum variance reduction}; this algorithm is the basis of our quantum speedups for stochastic optimization problems. We then discuss the application of quantum variance reduction for stochastic convex optimization by presenting two quantum algorithms in \sec{Q-AC-SA} and \sec{QSCPM}, respectively. Technically, it is possible to obtain our result in \sec{QSCPM} just using the quantum mean estimation routine of \cite{cornelissen2022near} rather then quantum variance reduction, however our use of quantum variance reduction facilitates our presentation. Furthermore, we present quantum algorithms for non-convex optimization based on quantum variance reduction in \sec{nonconvex}. Finally, in \sec{lower} we prove quantum lower bounds for quantum variance reduction and quantum stochastic convex optimization that establish the optimality of our algorithms, and conclude the paper in \sec{conclusion}.

%% file: sec-var-red.tex
\section{Quantum variance reduction}\label{sec:QVR}

We obtain our quantum speedups for stochastic optimization problems by proposing and developing new algorithms for a new problem which we call the \emph{quantum variance reduction problem}.

\begin{problem}[Variance reduction]\label{prob:variance-reduction}
For a $d$-dimensional random variable $X$ with $\Var[X]\leq L^2$ and some $\hat{\sigma}\geq 0$, suppose we are given access to its quantum sampling oracle $O_X$ defined in \defn{OX}. The goal is to output an unbiased estimate $\hat{\mu}$ of $\mu\coloneqq\E[X]$ satisfying $\E\|\hat{\mu}-\mu\|^2\leq \hat{\sigma}^2$.
\end{problem}

Classically, \prob{variance-reduction} can be solved by averaging $\mO(L^2/\hat{\sigma}^2)$ samples of $X$; this query complexity is optimal among classical algorithms~\cite{lugosi2019mean}. However, if we have access to the \textit{quantum sampling oracle} defined in \defn{OX}, \cite{cornelissen2022near} showed that a (possibly) biased estimate $\hat{\mu}$ with error $\norm{\hat{\mu} - \mu} \leq \hat{\sigma}$ can be computed using $\tO(L\sqrt{d} \hat{\sigma}^{-1})$ queries. 

\begin{lemma}[{\cite[Theorem 3.5] {cornelissen2022near}}]\label{lem:quantum-multivariate-mean-estimator}  
Given access to the quantum sampling oracle $O_X$, for any $\hat{\sigma},\delta\geq 0$ there is a procedure $\mathtt{QuantumMeanEstimation}\left(X,\hat{\sigma},\delta\right)$ that 
uses $\tO(L\sqrt{d}\log(1/\delta)/\hat{\sigma})$ queries and outputs an estimate $\hat{\mu}$ of the expectation $\mu$ of any $d$-dimensional random variable $X$ satisfying $\Var[X]\leq L^2$ with error $\|\hat{\mu}-\mu\|\leq\hat{\sigma}\leq L$ and success probability $1-\delta$. 
\end{lemma}

\texttt{QuantumMeanEstimation}~\cite{cornelissen2022near} proceeds by introducing a directional mean function that reduces a multivariate mean estimation problem to a series of univariate mean estimation problem through quantum Fourier transform, which in the bounded norm case can be solved by quantum algorithms with a quadratic speedup using phase estimation. In terms of error rate, \texttt{QuantumMeanEstimation} in \lem{quantum-multivariate-mean-estimator} improves over any classical sub-Gaussian estimator when $\frac{\hat{\sigma}}{L}\geq\frac{1}{\sqrt{d}}$. However, its bias hinders its combination with various optimization algorithms assuming unbiased inputs, see e.g.,~\cite{lan2012optimal,allen2017katyusha,fang2018spider,fang2019sharp}. In this work we show how to carefully combine their algorithm with a classic multi-level Monte-Carlo (MLMC) technique from~\cite{blanchet2015unbiased,asi2021stochastic} to obtain an unbiased estimate $\hat{\mu}$ and success probability 1 with the same rate as~\cite{cornelissen2022near} and prove the following theorem. 
\begin{theorem}\label{thm:unbiased-estimator}
\algo{unbiased-Q} solves \prob{variance-reduction} using 
an expected $\tO(L d^{1/2} \hat{\sigma}^{-1})$ queries.
\end{theorem}

In the following \tab{VR-comparison} we provide a comparison between our result and previous works and in \sec{lower} we prove that our algorithm is optimal up to a poly-logarithmic factor. Notably, our algorithm does not depend on the detailed implementation of \texttt{QuantumMeanEstimation} but only its query complexity. Hence, we believe that our approach may also be useful in removing the bias of other quantum mean estimation algorithms with similar expressions of query complexities, e.g., quantum phase estimation~\cite{kitaev1995quantum,cornelissen2023sublinear} and quantum amplitude estimation~\cite{brassard2002quantum}. 

\begin{table}[ht]
\centering
\begin{tabular}{cccc}
\hline
Queries & Bias & Variance & Method \\ \hline \hline
$\hat{\sigma}^{-2}$ & 0 & $\hat{\sigma}^2$ & Classical Variance Reduction \\ \hline
$d^{1/2} \hat{\sigma}^{-1}$ & $\hat{\sigma}$ & $\hat{\sigma}^2$ & Quantum Multivariate Mean estimation~\cite{cornelissen2022near} \\ \hline
$\log^2\left(1/\gamma\right) \hat{\sigma}^{-1}$ & $\gamma$ & $\hat{\sigma}^2$ & One-dimensional Quantum Mean estimation~\cite{cornelissen2023sublinear} \\ \hline
$d^{1/2} \hat{\sigma}^{-1}$ & 0  & $\hat{\sigma}^2$  & 
\textbf{Our Result}
(\thm{unbiased-estimator})
\\ \hline
\end{tabular}
\caption{Comparison between different methods for variance reduction in the case of $L=1$ and $\hat{\sigma} \in (0, 1)$. The $\tO$ symbol was omitted in the ``Queries'' column.
}
\label{tab:VR-comparison}
\end{table}

Our algorithm consists of two components. First, we show that we can obtain a variant of the quantum mean estimation algorithm, denoted as $\mathtt{QuantumMeanEstimation}^{+}\left(X,\hat{\sigma}\right)$, that outputs a low variance estimate with probability 1. This procedure compares the outcomes of the quantum estimator and a classical estimate, and in the event of significant disparity, generates a new independent classical estimate.\footnote{At \lin{X2} and \lin{X3} of \algo{qme+}, we can further refine our approach to obtain $X_2$ and $X_4$ by taking the same number of classical samples as we would have used in the \texttt{QuantumMeanEstimation} procedure, which can reduce the poly-logarithmic factor in the overall query complexity.}

\begin{algorithm2e}
\caption{$\mathtt{QuantumMeanEstimation}^{+}(X,\hat{\sigma})$}
	\label{algo:qme+}
	\LinesNumbered
	\DontPrintSemicolon
    \KwInput{Random variable $X$, target variance $\hat{\sigma}^2\leq L^2$}
    \KwParameter{$\delta=\hat{\sigma}^6/(4L)^6$, $D=\frac{\hat{\sigma}}{4} + \frac{16L^3}{\hat{\sigma}^2}$}
    \KwOutput{An estimate $\hat{\mu}$ of $\mu=\E[X]$ satisfying $\E\|\hat{\mu}-\mu\|^2\leq\hat{\sigma}^2$}
    Set $X_1\leftarrow$\texttt{QuantumMeanEstimation}$(X,\hat{\sigma}/4,\delta)$\;
    Randomly draw one classical sample $X_2$ of $X$ \label{lin:X2}\;
    \lIf{$\|X_1-X_2\|\leq D$}{\Return $X_1$}
    \Else{
        Randomly draw one classical sample $X_3$ of $X$\label{lin:X3}\;
     \Return $X_3$
    }
    
\end{algorithm2e}

Second, we use the MLMC technique~\cite{giles2015multilevel}, specifically a variant of the methods described in~\cite{blanchet2015unbiased,asi2021stochastic}, to carefully invoke the biased subroutine \texttt{QuantumMeanEstimation}$^+$ and compute the unbiased estimate. The algorithm, \algo{unbiased-Q}, simply invokes \texttt{QuantumMeanEstimation}$^+$ for three randomly chosen accuracies and combines them to obtain the result. Though there are an infinite number of possible accuuracies chosen, we show that the expectation, the variance, and expected number of queries are all suitable to prove \thm{unbiased-estimator}.

\begin{algorithm2e}
	\caption{Quantum variance reduction}
	\label{algo:unbiased-Q}
	\LinesNumbered
	\DontPrintSemicolon
    \KwInput{Random variable $X$, target variance $\hat{\sigma}^2$}
    \KwOutput{An unbiased estimate $\hat{\mu}$ of $\E[X]$ with variance at most $\hat{\sigma}^2$}
    Set $\tilde{\mu}_0\leftarrow$\texttt{QuantumMeanEstimation}$^+(X,\hat{\sigma}/10)$\;
	Randomly sample $j\sim\mathrm{Geom}\left(\frac{1}{2}\right)\in\N$\;
        $\tilde{\mu}_j\leftarrow$\texttt{QuantumMeanEstimation}$^+(X,2^{-3j/4}\hat{\sigma}/10)$\;
        $\tilde{\mu}_{j-1}\leftarrow$\texttt{QuantumMeanEstimation}$^+(X,2^{-3(j-1)/4}\hat{\sigma}/10)$\;
        $\hat{\mu}\leftarrow\tilde{\mu}_0+2^j(\tilde{\mu}_j-\tilde{\mu}_{j-1})$\label{lin:de-biasing}\;
        \Return $\hat{\mu}$\;
\end{algorithm2e}

We begin by analyzing any random process that broadly follows the structure of \algo{qme+}. The lemma shows how two unbiased, bounded variance random variables ($X_2$ and $X_3$ in the lemma) can be used to control the expected $\ell_2$ error of a random variable ($X_1$  in the lemma) that is close to the expectation of $X_2$ and $X_3$ with some probability. In \algo{qme+} $X_1$ corresponds to the output of \texttt{QuantumMeanEstimation} and then $X_2$ and $X_3$ are obtained by classic sampling. 

\begin{lemma}\label{lem:eliminating-failure}
Let $X_1,X_2,X_3 \in \R^d$ be independent random variables where $\E X_2 = \E X_3 = \mu$, $\Var[X_2] \leq L^2$, and $\Var[X_3] \leq L^2$, and $\norm{X_1 - \mu} \leq \hat{\sigma}$ with probability at least $1 - \delta$. Further, let $Z$ be the random variable that is $X_1$ if $\norm{X_1 - X_2} \leq D$ and $X_3$ otherwise where $D > \hat{\sigma}$. Then, 
\[
\E \norm{Z - \mu}^2 \leq  
\hat{\sigma}^2 + L \left( 
\frac{L^2 - \hat{\sigma}^2}{D - \hat{\sigma}}
\right)
+ \delta (2 D^2 + 3 L^2)\,.
\]
Consequently, $\E \norm{Z - \mu}^2 \leq 13 \sigma^2$ when $D = \hat{\sigma} + L^3 / \hat{\sigma}^2$ and $\delta \leq \hat{\sigma}^6/L^6$.
\end{lemma}

\begin{proof}
Let $S$ denote the event that $\norm{X_1 - \mu}_2 \leq \hat{\sigma}$ and let $T$ denote the event that $\norm{X_1 - X_2}_2 \leq D$. Then by the law of total expectation
\begin{align*}
\E \norm{Z - \mu}^2
&= \Pr\left[ S \right] \E \left[ \norm{Z - \mu}_2^2 ~ | ~ S \right] + \Pr\left[ \text{not } S \right] \E \left[ \norm{Z - \mu}_2^2 ~ | ~ \text{not } S \right] \\
&\leq 
\E \left[ \norm{Z - \mu}^2 ~ | ~ S \right] + \delta \cdot \E \left[ \norm{Z - \mu}^2 ~ | ~ \text{not } S \right] 
\end{align*}
Where we used that $\Pr[S] \leq 1$ and $\Pr[\text{not } S] \leq \delta$. We will prove the lemma by leveraging this inequality and upper bounding both $\E \left[ \norm{Z - \mu}_2^2 ~ | ~ S \right]$ and $\E \left[ \norm{Z - \mu}_2^2 ~ | ~ \text{not } S \right]$.

First, we upper bound $\E \left[ \norm{Z - \mu}_2^2 ~ | ~ S \right]$. Again by the law of total expectation and the definition of $Z$ we have that
\begin{align}\label{eqn:Z-bound}
\E \left[ \norm{Z - \mu}_2^2 ~ | ~ S \right]
&= \Pr\left[ T ~ | ~ S \right] \E \left[ \norm{Z - \mu}_2^2 ~ | ~ T \text{ and } S \right] + \Pr\left[ \text{not } T ~ | ~ S \right] \E \left[ \norm{Z - \mu}_2^2 ~ | ~ T \text{ and not } S \right]\nonumber \\
&\leq \Pr\left[ T ~ | ~ S \right] \cdot \hat{\sigma}^2 + \Pr\left[ \text{not } T ~ | ~ S \right] L^2 \,.
\end{align}
Here we use that if $T$ and $S$ both hold then $Z = X_1$ and $\norm{X_1 - \mu}\leq \hat{\sigma}$ and that if $S$ holds and not $T$ then $Z = X_3$ and $X_3$ is independent of $S$ and $T$ with variance $L^2$. Now, if $\norm{X_2 - \mu}_2 \leq D - \hat{\sigma}$ then when $S$ holds by triangle inequality 
\[
\norm{X_1 - X_2} \leq \norm{X_1 - \mu} + \norm{X_2 - \mu}
\leq D\,.
\]
Applying this bound and Chebyshev inequality to $X_2$ (and using that $X_2$ is independent of $S$) we have that
\[
\Pr [T ~ | ~ S]
\geq \Pr[ \norm{X_2 - \mu}_2 \leq D - \hat{\sigma} ~ | ~ S ]
= \Pr[ \norm{X_2 - \mu}_2 \leq D - \hat{\sigma} ]
\geq 1 - \frac{L}{D - \hat{\sigma}}\,.
\]
Since $\hat{\sigma} < L$, combined with \eqn{Z-bound} this implies that
\[
\E \left[ \norm{Z - \mu}_2^2 ~ | ~ S \right]
\leq \left(1 - \frac{L}{D - \hat{\sigma}}\right) \hat{\sigma}^2
+ \frac{L^3}{D - \hat{\sigma}}
= \hat{\sigma}^2 + L \left( 
\frac{L^2 - \hat{\sigma}^2}{D - \hat{\sigma}}
\right)
\]

Next, we upper bound $\E \left[ \norm{Z - \mu}_2^2 ~ | ~ \text{not } S \right]$. Note that if $T$ holds then 
\[
\norm{Z - \mu}^2 \leq 2 \norm{Z - X_2}^2 + 2 \norm{X_2 - \mu}^2
\leq 2 D^2 + 2 \norm{X_2 - \mu}^2\,.
\]
Since $Z = X_3$ if $T$ does note hold, we then see that in either case
\[
\norm{Z - \mu}^2 \leq 2 D^2 + 2 \norm{X_2 - \mu}^2 + \norm{X_3 - \mu}^2
\]
Since, $X_2$ and $X_3$ are independent of $S$ with variance $L$ this therefore implies that 
\[
\E \left[ \norm{Z - \mu}_2^2 ~ | ~ \text{not } S \right]
\leq \E\left[ 2 D^2 + 2 \norm{X_2 - \mu}^2 + \norm{X_3 - \mu}^2 \right] = 2 D^2 + 3L^2,
\]
which leads to
\begin{align*}
\E \norm{Z - \mu}^2 
&\leq 
\E \left[ \norm{Z - \mu}^2 ~ | ~ S \right] + \delta \cdot \E \left[ \norm{Z - \mu}^2 ~ | ~ \text{not } S \right] \\
&\leq \hat{\sigma}^2 + L \left( 
\frac{L^2 - \hat{\sigma}^2}{D - \hat{\sigma}}
\right)+\delta(2D^2+3L^2).
\end{align*}
\end{proof}

\begin{corollary}\label{cor:qme+}
Given access to the quantum sampling oracle $O_X$ of a $d$-dimensional random variable $X$, for any $\hat{\sigma},\delta\geq 0$, the procedure $\mathtt{QuantumMeanEstimation}^+\left(X,\hat{\sigma}\right)$ (\algo{qme+}) uses $\tO(L\sqrt{d}\hat{\sigma}^{-1})$ queries and outputs an estimate $\hat{\mu}$ of the expectation $\mu=\E[X]$ satisfying $\E\|\hat{\mu}-\mu\|^2\leq\hat{\sigma}^2$.
\end{corollary}
\begin{proof}
By \lem{eliminating-failure}, the output $\hat{\mu}$ of $\mathtt{QuantumMeanEstimation}^+\left(X,\hat{\sigma}\right)$ (\algo{qme+}) satisfies 
\begin{align}
\E\|\hat{\mu}-\mu\|^2\leq \hat{\sigma}^2.
\end{align}
As for the number of queries to $O_X$, note that it uses 
\begin{align*}
\tO\left(\frac{L\sqrt{d}\log(1/\delta)}{\hat{\sigma}}\right)=\tO\left(\frac{L\sqrt{d}}{\hat{\sigma}}\log\left(\frac{L^6}{\hat{\sigma}^6}\right)\right)=\tO\left(\frac{L\sqrt{d}}{\hat{\sigma}}\right)
\end{align*}
queries to prepare $X_1$. Moreover, we can prepare the classical samples $X_2$ and $X_3$ by measuring the quantum state $O_X\ket{0}$ in the computational basis, which takes $1$ query to $O_X$. Hence, the overall query complexity equals $\tO(L\sqrt{d}\hat{\sigma}^{-1})$.
\end{proof}

\begin{proof}[Proof of \thm{unbiased-estimator}]
The structure of our proof is similar to the proof of \cite[Proposition 1]{asi2021stochastic}. Observe that the output $\hat{\mu}$ of \algo{unbiased-Q} can be expressed as 
\begin{align}\label{eqn:MLMC-estimator}
\hat{\mu}=\tilde{\mu}_0+2^J(\tilde{\mu}_J-\tilde{\mu}_{J-1}),\qquad J\sim\mathrm{Geom}\Big(\frac{1}{2}\Big)\in\N.
\end{align}
Then we have
\begin{align*}
\mathbb{E}[\hat{\mu}]=\E[\tilde{\mu}_0]+\sum_{j=1}^{\infty}\Pr\{J=j\}2^j(\E[\tilde{\mu}_j]-\E[\tilde{\mu}_{j-1}])
=\E[\tilde{\mu}_{\infty}]=\mu,
\end{align*}
given that $\Pr\{J=j\}=2^{-j}$. As for the variance, we have
\begin{align*}
\E\|\hat{\mu}-\mu\|^2\leq 2\E\|\hat{\mu}-\tilde{\mu}_0\|^2+2\E\|\tilde{\mu}_0-\mu\|^2,
\end{align*}
where
\begin{align*}
\mathbb{E}\|\hat{\mu}-\tilde{\mu}_0\|^2
=\sum_{j=1}^{\infty}\Pr(J=j)2^{2j}\E\|\tilde{\mu}_j-\tilde{\mu}_{j-1}\|^2
=\sum_{j=1}^{\infty}2^j\E\|\tilde{\mu}_j-\tilde{\mu}_{j-1}\|^2,
\end{align*}
and for each $j$ we have 
\begin{align*}
\E\|\tilde{\mu}_j-\tilde{\mu}_{j-1}\|^2
\leq 2\E\|\tilde{\mu}_j-\mu\|^2+2\E\|\tilde{\mu}_{j-1}-\mu\|^2.
\end{align*}
By \cor{qme+}, 
\begin{align*}
\E\|\tilde{\mu}_j-\mu\|^2\leq\frac{\hat{\sigma}^2}{100\cdot 2^{3j/2}},\quad\forall j\geq 0,
\end{align*}
which leads to 
\begin{align*}
\E\|\tilde{\mu}_j-\tilde{\mu}_{j-1}\|^2\leq \frac{\sigma^2}{50\cdot 2^{3(j-1)/2}}+\frac{\sigma^2}{50\cdot 2^{3j/2}}\leq\frac{\hat{\sigma}^2}{10\cdot 2^{3j/2}},
\end{align*}
and
\begin{align*}
\E\|\hat{\mu}-\tilde{\mu}_0\|^2=\frac{\hat{\sigma}^2}{10}\sum_{j=1}^{\infty}\frac{1}{2^{j/2}}\leq \frac{1}{3} \hat{\sigma}^2\,.
\end{align*}
Hence,
\begin{align*}
\E\|\hat{\mu}-\mu\|^2\leq 2\E\|\hat{\mu}-\tilde{\mu}_0\|^2+2\E\|\tilde{\mu}_0-\mu\|^2\leq \hat{\sigma}^2,
\end{align*}
given that $\E\|\tilde{\mu}_0-\mu\|^2\leq\hat{\sigma}^2/100$ by \cor{qme+}. Moreover, the expected number of queries is
\begin{align*}
\tO\left(\frac{L\sqrt{d}\log d}{\hat{\sigma}}\right)\cdot\left(1+\sum_{j=1}^{\infty}\Pr\{J=j\}\cdot\left(2^{3j/4}+2^{3(j-1)/4}\right)\right)=\tO\left(\frac{L\sqrt{d}}{\hat{\sigma}}\right).
\end{align*}
\end{proof}

%% file: sec-lower.tex
\section{Lower bounds}\label{sec:lower}

In this section we present two quantum lower bounds for solving quantum variance reduction (\prob{variance-reduction}) and stochastic convex optimization (\prob{SCO}), respectively. 

\subsection{Quantum lower bound for variance reduction}

We first establish the following quantum lower bound for the variance reduction problem (\prob{variance-reduction}) which shows that our \algo{unbiased-Q} is optimal up to a poly-logarithmic factor when $\hat{\sigma}=\mO\left(d^{-1/2}\right)$. By Markov's inequality, \prop{QVR-lower} equivalently states that any quantum algorithm that solves \prob{variance-reduction} must make an expected $\Omega(L\sqrt{d}\hat{\sigma}^{-1})$ queries. This matches our algorithmic result provided in \thm{unbiased-estimator}, up to a poly-logarithmic factor. 

\begin{proposition}\label{prop:QVR-lower} 
There is a constant $\alpha$ such that for any $\hat{\sigma}\leq\frac{\sigma}{\alpha\sqrt{d}}$, any quantum algorithm that solves \prob{variance-reduction} with success probability at least $2/3$ must make at least $\Omega(L\sqrt{d}\hat{\sigma}^{-1})$ queries in the worst case.
\end{proposition}

Note that \prob{variance-reduction} is strictly harder than the multivariate mean estimation problem considered in~\cite{cornelissen2022near}. In the latter, the objective is to produce an estimate of the expected value within a bounded $\ell_2$ error, without the requirement of unbiasedness. A query lower bound for solving quantum mean estimation was established in~\cite{cornelissen2022near}.

\begin{lemma}[{\cite[Theorem 3.8]{cornelissen2022near}}]\label{lem:mean-estimation-lower}
Let $\sigma\geq 0$ and denote $P_\sigma$ to be the set of all $d$-dimensional quantum random variables with covariance matrix $\Sigma$ such that $\Tr(\Sigma)=\sigma^2$. Suppose every $X\in P_\sigma$ is indexed by some random seeds $\omega$, i.e., 
\begin{align*}
X=X(\omega),\quad\omega\sim p_{\Omega},
\end{align*}
where $p_{\Omega}$ denotes the probability distribution of the random seed $\omega$. Then there exists a constant $\alpha$ such that for any $n>\alpha d$ and any quantum algorithm that uses at most $n$ queries to the following two oracles
\begin{align*}
O_{\Omega}\ket{0}\to\int_\omega\sqrt{p_\Omega(\omega)\d\omega}\ket{\omega}
\qquad
\text{ and }
\qquad
\mathcal{B}_X\ket{\omega}\ket{0}\to\ket{\omega}\ket{X(\omega)},
\end{align*}
there exists an instance $X\in P_\sigma$ such that the quantum algorithm returns a mean estimate $\tilde{\mu}$ of the mean $\mu$ of $X$ that satisfies 
\begin{align*}
\|\tilde{\mu}-\mu\|\geq\Omega\left(\frac{\sigma\sqrt{d}}{n}\right)
\end{align*}
with probability at least $2/3$.
\end{lemma}
By employing an oracle reduction argument, we obtain our lower bound result based on \lem{mean-estimation-lower}.
\begin{proof}[Proof of \prop{QVR-lower}]
We first show that one query to the quantum sampling oracle $O_X$ defined in \defn{OX} can be implemented using one query to $O_{\Omega}$ and one query to $\mathcal{B}_X$. In particular, observe that
\begin{align*}
\mathcal{B}_XO_\Omega\ket{0}\otimes\ket{0}=\int_\omega\sqrt{p_{\Omega}(\omega)\d\omega}\ket{\omega}\otimes\ket{X(\omega)}.
\end{align*}
By measuring the first register in the computational basis spanned by all the random seeds $\{\ket{\omega}\}$, we can obtain the desired output of $O_X$. Hence, if there exists a quantum algorithm making at most $\mO(L\sqrt{d}/\hat{\sigma})$ queries to the oracle $O_X$ that solves \prob{variance-reduction} with success probability at least $2/3$, there also exists a quantum mean estimation algorithm that uses at most $\mO(L\sqrt{d}/\hat{\sigma})$ queries to $O_\Omega$ and $\mathcal{B}_X$ and for any $X\in P_\sigma$ it outputs an estimate $\tilde{\mu}$ of the mean $\mu$ of $X$ that satisfies $\|\tilde{\mu}-\mu\|\leq \mO(\hat{\sigma})$ with probability at least $2/3$, which contradicts to \lem{mean-estimation-lower}.
\end{proof}

\subsection{Quantum lower bound for stochastic convex optimization}

Next, we establish the following quantum lower bounds for stochastic convex optimization (\prob{SCO}) in the low-dimension regime and the high-dimension regime, respectively, which show that our quantum stochastic cutting plane method in \sec{QSCPM} is optimal up to a poly-logarithmic factor when the dimension $d$ is a constant, and there is no quantum speedup over SGD when $d\geq\Omega\left(\epsilon^{-2}\right)$.

\begin{theorem}\label{thm:SCO-lower}
For any $\epsilon\leq\frac{RL}{100\sqrt{d}}$, any quantum algorithm that solves \prob{SCO} with success probability at least $2/3$ must make at least $\Omega(\sqrt{d}RL/\epsilon)$ queries in the worst case. For any $\frac{RL}{100\sqrt{d}}\leq\epsilon\leq 1$, any quantum algorithm that solves \prob{SCO} with success probability at least $2/3$ must make at least $\Omega(R^2L^2/\epsilon^2)$ queries in the worst case.
\end{theorem}

The proof of \thm{SCO-lower} follow a similar approach to the previous proof of the quantum mean estimation lower bound~\cite{cornelissen2022near}. Specifically, our results are derived through a reduction to a specific multivariate mean estimation problem. In this problem, all vectors are sampled from a set of orthonormal bases and there might be duplicated vectors. We choose to reduce to this particular problem instead of directly reducing to mean estimation problem because the construction of our hard instances require a tighter control over the norm of the vectors as well as the norm of their expected value. This problem can be equivalently represented as the following composition problem, as defined in~\cite{cornelissen2022near}. 

\begin{problem}[$\search^N\circ \parity^M$]\label{prob:parity-search}
Let $N,M\geq 1$ be two integers. Let $\mathcal{A}_{N,M}$ denote the set of matrices $A\in\{0,1\}^{N\times M}$ such that $\lfloor N/2\rfloor$ rows have Hamming weights $\lfloor M/2\rfloor$, and the other rows have Hamming weights $\lfloor M/2\rfloor+1$. Define the vector $\vect{b}^{(A)}$ such that,
\begin{align}\label{eqn:bA-defn}
b_i^{(A)}=
\begin{cases}
0,\quad\text{if the }i\text{-th row of }A\text{ has Hamming weight }\lfloor M/2\rfloor,\\
1,\quad\text{if the }i\text{-th row of }A\text{ has Hamming weight }\lfloor M/2\rfloor+1,
\end{cases}
\end{align}
for each $i\in[N]$. Then, the $\search^N\circ \parity^M$ problem consists of finding a vector $\widetilde{\vect{b}}\in\R^N$ that minimizes $\|\widetilde{\vect{b}}-\vect{b}^{(A)}\|_1$ given a quantum oracle 
\begin{align}\label{eqn:OA-defn}
O_A\ket{i,j}\to(-1)^{A_{ij}}\ket{i,j}.
\end{align}
\end{problem}

\cite{cornelissen2022near} showed that $\Omega(NM)$ queries are needed to approximate the vector $\vect{b}^{(A)}$ with small error.

\begin{lemma}[{\cite[Lemma 3.6]{cornelissen2022near}}]\label{lem:search-parity}
There exists a constant $\alpha>1$ such that, for any quantum algorithm for the $\search^N\circ\parity^M$ problem that uses at most $NM/\alpha$ queries there exists an $A\in\mathcal{A}_{N,M}$ such that this algorithm returns a vector $\widetilde{\vect{b}}$ satisfying $\big\|\widetilde{\vect{b}}-\vect{b}^{(A)}\big\|\geq\sqrt{N}/2$ with probability at least $2/3$.
\end{lemma}

We establish our quantum lower bounds for stochastic convex optimization by establishing a correspondence between the $\search^N\circ\parity^M$ problem and \prob{SCO}. Specifically, for any input $A\in\mathcal{A}_{N\times M}$ to the $\search^N\circ\parity^M$ problem, we design the following convex function whose optimal point corresponds to the solution of the $\search^N\circ\parity^M$ problem under certain sets of parameters,
\begin{align}\label{eqn:barf-defn-combined}
\bar{f}^A(\x)\coloneqq\underset{i,j}{\E}[f_{i,j}(\x)],
\end{align}
where
\begin{align*}
f^A_{i,j}(\x)\coloneqq
-\frac{1}{3}\left\<\x,\g_{i,j}\right\>+\frac{2L}{3}\cdot\max\left\{0,\|\x\|-\frac{R}{2}\right\}
\end{align*}
for some $\g_{i,j}\in\R^d$ satisfying $\|\g_{i,j}\|\leq L$ for all $i\in[N]$ and $j\in[M]$.

\begin{lemma}\label{lem:barf-properties-combined}
Denote $\bar{\g}\coloneqq\underset{i,j}{\E}[\g_{i,j}]$. Then, the function $\bar{f}^A$ defined in \eqn{barf-defn-combined} has the following properties if $\bar{\g}$ is a non-zero vector.
\begin{enumerate}
\item $\bar{f}^A$ is convex.
\item $\bar{f}^A$ is minimized at $\x^*=\frac{R}{2}\frac{\bar{\g}}{\|\bar{\g}\|}$.
\item Every $\epsilon$-optimum $\x$ of $\bar{f}^A$ satisfies
\begin{align*}
\left\<\frac{R}{2}\cdot\frac{\x}{\|\x\|},\x^*\right\>\geq 1-\frac{R\epsilon}{2\|\bar{\g}\|}
\end{align*}
\item For 
\begin{align}\label{eqn:hatg-defn}
\hat{\g}_{i,j}(\x)\coloneqq -\frac{1}{3}\cdot\g_{i,j}+\frac{2L}{3}\cdot\max\left\{\|\x\|-\frac{R}{2},0\right\}\cdot\frac{\x}{\|\x\|}
\end{align}
it is the case that
\begin{align*}
\underset{i,j}{\E}[\hat{\g}_{i,j}(\x)]\in\partial\bar{f}^A(\x)
\quad\text{ and }\quad\|\hat{\g}_{i,j}(\x)\|\leq L,\qquad\forall\x\in\R^d.
\end{align*}
\end{enumerate}
\end{lemma}

\begin{proof}
Since  $\bar{f}^A$ is the sum of a linear function and a maximum function, both of which are convex, $\bar{f}^A$ itself is convex.

Observe that the minimum of $\bar{f}^A$ cannot be achieved when $\|\x\|>R/2$ for otherwise the function value can be further decreased when moving towards $\0$ given that 
\begin{align*}
\|\bar{\g}\|=\left\|\frac{\sum_{i,j}\g_{i,j}}{NM}\right\|\leq\frac{\sum_{i,j}\|\g_{i,j}\|}{NM}\leq L.
\end{align*}
When $\|\x\|\leq R/2$, $\bar{f}^A$ can be expressed as
\begin{align*}
\bar{f}^A(\x)=-\frac{1}{3}\left\<\x,\bar{\g}\right\>,
\end{align*}
which is minimized at
\begin{align*}
\frac{R}{2}\cdot\frac{\bar{\g}}{\|\bar{\g}\|}=\frac{R\vect{b}^{(A)}}{\sqrt{2d}}.
\end{align*}
For any $\x$ that is an $\epsilon$-optimum of $\bar{f}^A$, we define 
\begin{align*}
\x'=\begin{cases}
\x,\quad\|\x\|\leq\frac{R}{2},\\
\frac{R}{2}\cdot\frac{\x}{\|\x\|},\quad\text{otherwise},
\end{cases}
\end{align*}
which satisfies $f(\x')\leq f(\x)$ and is thus also an $\epsilon$-optimum of $f$. Then we can derive that
\begin{align*}
-\frac{1}{3}\left\<\x',\bar{\g}\right\>+\frac{1}{3}\left\<\x^*,\bar{\g}\right\>\leq\epsilon,
\end{align*}
and
\begin{align*}
\left\<\x',\x^*\right\>\geq \left\<\x^*,\x^*\right\>-\epsilon\cdot\frac{\|\x^*\|}{\|\bar{\g}\|}\geq 1-\frac{R\epsilon}{2\|\bar{\g}\|},
\end{align*}
which leads to
\begin{align*}
\left\<\frac{R}{2}\cdot\frac{\x}{\|\x\|},\x^*\right\>\geq 1-\frac{R\epsilon}{2\|\bar{\g}\|}
\end{align*}
given that $\frac{R}{2}\frac{\x}{\|\x\|}\geq\|\x'\|$.
As for the last entry, note that
\begin{align*}
\underset{i,j}{\E}[\hat{\g}_{i,j}(\x)]&=-\frac{1}{3}\cdot\underset{i,j}{\E}\g_{i,j}+\frac{2L}{3}\cdot\max\left\{\|\x\|-\frac{R}{2}\right\}\cdot\frac{\x}{\|\x\|}\\
&=-\frac{1}{3}\cdot\bar{\g}+\frac{2L}{3}\cdot\max\left\{\|\x\|-\frac{R}{2}\right\}\cdot\frac{\x}{\|\x\|},
\end{align*}
which is the (sub-)gradient of $\bar{f}^A$ at $\x$. Moreover, for any $i,j$ we have
\begin{align*}
\|\hat{\g}_{i,j}(\x)\|\leq\frac{1}{3}\|\g_{i,j}\|+\frac{2L}{3}\leq L.
\end{align*}
\end{proof}

Equipped with \lem{barf-properties-combined}, we first prove our quantum lower bound for \prob{SCO} in the low-dimensional regime by encoding the $\search\circ\parity$ problem into the task of finding an $\epsilon$-optimal point of $\bar{f}^A$ defined in \eqn{barf-defn-combined}. 

\begin{proof}[Proof of \thm{SCO-lower} when $\epsilon\leq\frac{RL}{100\sqrt{d}}$]
Our proof proceeds by contradiction. Denote $n=\frac{RL\sqrt{d}}{100\alpha\epsilon}$. Assume for simplicity that $d$ is even, and $\alpha n$ is an even multiple of $d$ (the other cases can be handled by padding arguments). Then, we show that $\search^d\circ\parity^{\alpha n/d}$ on this instance $A$ can be solved by finding an $\epsilon$-optimum of $\bar{f}^A$ defined in \eqn{barf-defn-combined}. In particular, we define the set of vectors $\{\g_{i,j}\}$ to be
\begin{align}\label{eqn:g-low}
\g_{i,j}=\frac{\alpha L n}{\sqrt{(2\alpha n)^2-d^2}}(-1)^{1+A_{i,j}}\e_i,
\end{align}
where $i\in[d]$, $j\in [\alpha n/d]$, and $\e_i\in\R^d$ is the $i$-th indicator vector. Consider the function $\bar{f}^A$ defined in \eqn{barf-defn-combined} with the set of vectors $\{\g_{i,j}\}$ having the values in \eqn{g-low}, we have
\begin{align*}
\bar{\g}=\underset{i,j}{\E}[\g_{i,j}]
&=\frac{1}{\alpha n}\cdot\frac{\alpha L n}{\sqrt{(2\alpha n)^2-d^2}}\sum_{i=1}^d\e_i \sum_{j=1}^{\alpha n/d}(-1)^{1+A_{i,j}}\\
&=\frac{2L}{\sqrt{(2\alpha n)^2-d^2}}\sum_{i=1}^d b_i^{(A)}=\frac{2L}{\sqrt{(2\alpha n)^2-d^2}}\vect{b}^{(A)},
\end{align*}
where the vector $\vect{b}^{(A)}$ is defined in \eqn{bA-defn}. By \lem{barf-properties-combined}, every $\epsilon$-optimal point $\x$ of $\bar{f}^A$ satisfies 
\begin{align*}
\left\<\frac{R}{2}\frac{\x}{\|\x\|},\x^*\right\>\geq 1-\frac{R\epsilon}{2\|\bar{\g}\|},
\end{align*}
by which we can derive that
\begin{align*}
\left\<\frac{R}{2}\frac{\x}{\|\x\|},\vect{b}^{(A)}\right\>\geq \frac{\|\vect{b}^{(A)}\|}{\|\x^*\|}\cdot\left(1-\frac{R\epsilon}{2\|\bar{\g}\|}\right)\geq \frac{R}{2}\sqrt{\frac{d}{2}}-\frac{3\alpha n\epsilon}{L}=\frac{R\sqrt{d}}{4},
\end{align*}
which leads to
\begin{align*}
\left\<\sqrt{\frac{d}{2}}\cdot\frac{\x}{\|\x\|},\vect{b}^{(A)}\right\>\geq\frac{d}{2\sqrt{2}}.
\end{align*}
Given that
\begin{align*}
\left\|\sqrt{\frac{d}{2}}\cdot\frac{\x}{\|\x\|}\right\|=\left\|\vect{b}^{(A)}\right\|=\sqrt{\frac{d}{2}},
\end{align*}
we have
\begin{align*}
\left\|\sqrt{\frac{d}{2}}\cdot\frac{\x}{\|\x\|}-\vect{b}^{(A)}\right\|\leq\frac{\sqrt{d}}{2},
\end{align*}
indicating that we can obtain an estimate $\widetilde{\vect{b}}=\sqrt{\frac{d}{2}}\frac{\x}{\|\x\|}$ of $\vect{b}^{(A)}$ with $\ell_2$-error at most $\sqrt{d}/2$ by obtaining an $\epsilon$-optimal point of $\bar{f}^A$.

Next, we show that we can implement a quantum stochastic gradient oracle of $\bar{f}^A$ using only one query to the oracle $O_A$, given that one can use one query to $O_A$ to implement the following oracle
\begin{align*}
O^A_{\g}\ket{i}\ket{j}\ket{0}\to\ket{i}\ket{j}\ket{\g_{i,j}}
\end{align*}
and use one query to $O^A_{\g}$ to implement 
\begin{align*}
O^A_{\hat{\g}}\ket{\x}\ket{i}\ket{j}\ket{0}\to\ket{\x}\ket{i}\ket{j}\ket{\hat{\g}_{i,j}(\x)},
\end{align*}
where $\hat{\g}_{i,j}(\x)$ is defined in \eqn{hatg-defn}. Then by applying $O^A_{\hat{\g}}$ to the state
\begin{align*}
\frac{1}{\sqrt{\alpha n}}\ket{\x}\sum_{i\in[d]}\sum_{j\in[\alpha n/d]}\ket{i}\ket{j}\ket{0}
\end{align*}
and uncompute the second and the third register, we construct the following stochastic gradient oracle $O_{\hat{\g}}$ of $\bar{f}^A$.
\begin{align*}
O_{\hat{\g}}\ket{\x}\otimes\ket{0}\to\frac{1}{\sqrt{\alpha n}}\ket{\x}\otimes\sum_{i\in[d]}\sum_{j\in[\alpha n/d]}\ket{\hat{\g}_{i,j}(\x)}.
\end{align*}
Hence, any quantum algorithm that can find an $\epsilon$-optimum of $\bar{f}^A$ using $T$ queries to $O_{\hat{\g}}$ can be transformed into a quantum algorithm for the $\search^d\circ\parity^{\alpha n/d}$ problem that uses $T$ queries to the oracle $O_A$ defined in \eqn{OA-defn} and returns an estimate $\widetilde{\vect{b}}$ with $\ell_2$-error at most $\sqrt{d}/2$. Then by \lem{search-parity}, for any quantum algorithm making less than
\begin{align*}
d\cdot\frac{\alpha n}{d}\cdot\frac{1}{\alpha}=\frac{RL}{2\epsilon\alpha}\sqrt{\frac{d}{2}}
\end{align*}
queries to the quantum stochastic gradient oracle $O_{\hat{\g}}$, there exists an $A\in\mathcal{A}_{d,\alpha n/d}$ and corresponding $\bar{f}^A$ such that the output of the algorithm is not an $\epsilon$-optimum of $\bar{f}^A$ with probability at least $2/3$.
\end{proof}

Similarly, we can obtain our quantum lower bound in the high-dimensional regime by encoding the $\search\circ\parity$ problem into the task of finding an $\epsilon$-optimal point of $\bar{f}^A$ defined in \eqn{barf-defn-combined}. 

\begin{proof}[Proof of \thm{SCO-lower} when $\frac{RL}{100\sqrt{d}}\leq\epsilon\leq 1$]
Our proof proceeds by contradiction. Denote $n=\frac{R^2L^2}{10000\alpha\epsilon^2}$. Consider an instance $A\in\mathcal{A}_{\alpha n,1}$ of the $\search^{\alpha n}\circ\parity^1$ problem where $d$ is a power of 2 and $\alpha$ is even. Then, we show that $\search^d\circ\parity^{\alpha n/d}$ on this instance $A$ can be solved by finding an $\epsilon$-optimum of $\bar{f}^A$ defined in \eqn{barf-defn-combined}. In particular, we define the set of vectors $\{\g_{i,j}\}$ where $j\equiv 1$ to be
\begin{align}\label{eqn:g-high}
\g_{i,1}=\alpha LA_{i,1}\sqrt{\frac{n}{2\alpha^2n-2\alpha}}\e_i,
\end{align}
where $i\in[\alpha n]$ and $\e_i\in\R^d$ is the $i$-th indicator vector. Consider the function $\bar{f}^A$ defined in \eqn{barf-defn-combined} with the set of vectors $\{\g_{i,1}\}$ having the values in \eqn{g-high}, we have
\begin{align*}
\bar{\g}=\underset{i}{\E}[\g_{i,1}]
=\frac{\alpha L}{\alpha n}\sqrt{\frac{n}{2\alpha^2n-2\alpha}}\sum_{i=1}^dA_{i,1}\e_i
=\frac{L}{n}\sqrt{\frac{n}{2\alpha^2n-2\alpha}}\vect{b}^{(A)},
\end{align*}
where the vector $\vect{b}^{(A)}$ is defined in \eqn{bA-defn}. By \lem{barf-properties-combined}, every $\epsilon$-optimal point $\x$ of $\bar{f}^A$ satisfies 
\begin{align*}
\left\<\frac{R}{2}\frac{\x}{\|\x\|},\x^*\right\>\geq 1-\frac{R\epsilon}{2\|\bar{\g}\|},
\end{align*}
by which we can derive that
\begin{align*}
\left\<\frac{R}{2}\frac{\x}{\|\x\|},\vect{b}^{(A)}\right\>\geq \frac{\|\vect{b}^{(A)}\|}{\|\x^*\|}\cdot\left(1-\frac{R\epsilon}{2\|\bar{\g}\|}\right)\geq \frac{R}{2}\sqrt{\frac{\alpha n}{2}}-\frac{6\alpha n\epsilon}{L}=\frac{R\sqrt{\alpha n}}{4},
\end{align*}
which leads to
\begin{align*}
\left\<\sqrt{\frac{\alpha n}{2}}\cdot\frac{\x}{\|\x\|},\vect{b}^{(A)}\right\>\geq\frac{\alpha n}{2\sqrt{2}}.
\end{align*}
Given that
\begin{align*}
\left\|\sqrt{\frac{\alpha n}{2}}\cdot\frac{\x'}{\|\x'\|}\right\|=\left\|\vect{b}^{(A)}\right\|=\sqrt{\frac{\alpha n}{2}},
\end{align*}
we have
\begin{align*}
\left\|\sqrt{\frac{\alpha n}{2}}\cdot\frac{\x}{\|\x\|}-\vect{b}^{(A)}\right\|\leq\frac{\sqrt{\alpha n}}{2}.
\end{align*}
indicating that we can obtain an estimate $\widetilde{\vect{b}}=\sqrt{\frac{\alpha n}{2}}\frac{\x}{\|\x\|}$ of $\vect{b}^{(A)}$ with $\ell_2$-error at most $\sqrt{d}/2$ by obtaining an $\epsilon$-optimal point of $\bar{f}^A$.

Next, we show that we can implement a quantum stochastic gradient oracle of $\bar{f}^A$ using only one query to the oracle $O_A$, given that one can use one query to $O_A$ to implement the following oracle
\begin{align*}
O^A_{\g}\ket{i}\ket{0}\to\ket{i}\ket{\g_{i,1}}
\end{align*}
and use one query to $O^A_{\g}$ to implement 
\begin{align*}
O^A_{\hat{\g}}\ket{\x}\ket{i}\ket{0}\to\ket{\x}\ket{i}\ket{\hat{\g}_{i,1}(\x)},
\end{align*}
where $\hat{\g}_{i,1}(\x)$ is defined in \eqn{hatg-defn}. Then by applying $O^A_{\hat{\g}}$ to the state
\begin{align*}
\frac{1}{\sqrt{\alpha n}}\ket{\x}\sum_{i\in[\alpha n]}\ket{i}\ket{0}
\end{align*}
and uncompute the second and the third register, we construct the following stochastic gradient oracle $O_{\hat{\g}}$ of $\bar{f}^A$.
\begin{align*}
O_{\hat{\g}}\ket{\x}\otimes\ket{0}\to\frac{1}{\sqrt{\alpha n}}\ket{\x}\otimes\sum_{i\in[\alpha n]}\ket{\hat{\g}_{i,1}(\x)}.
\end{align*}
Hence, any quantum algorithm that can find an $\epsilon$-optimum of $\bar{f}^A$ using $T$ queries to $O_{\hat{\g}}$ can be transformed into a quantum algorithm for the $\search^{\alpha n}\circ\parity^{1}$ problem that uses $T$ queries to the oracle $O_A$ defined in \eqn{OA-defn} and returns an estimate $\widetilde{\vect{b}}$ with $\ell_2$-error at most $\sqrt{d}/2$. Then by \lem{search-parity}, for any quantum algorithm making less than
\begin{align*}
\alpha n=\frac{R^2L^2}{100\epsilon^2}
\end{align*}
queries to the quantum stochastic gradient oracle $O_{\hat{\g}}$, there exists an $A\in\mathcal{A}_{\alpha n,1}$ and corresponding $\bar{f}^A$ such that the output of the algorithm is not an $\epsilon$-optimum of $\bar{f}^A$ with probability at least $2/3$.
\end{proof}

%% file: sec-conclusion.tex
\section{Conclusion}
\label{sec:conclusion}

We presented improved quantum algorithms for stochastic optimization. We developed a new technical tool which we call \emph{quantum variance reduction} and show how to use it to improve upon the query complexity for stochastic convex optimization and for critical point computation in smooth, stochastic, non-convex functions. Further, we provided lower bounds which establish both the optimality of our quantum variance reduction technique and of one of our stochastic convex optimization algorithms in low dimensions. A natural open problem suggested by our work is to establish the optimal complexity of the problems we study, e.g., stochastic convex optimization and stochastic non-convex optimization with quantum oracle access, in higher dimensions. We hope this paper fuels further study of these problems. 

\section*{Acknowledgement}
We thank Adam Bouland, Yair Carmon, Andr\'as Gily\'en, Yujia Jin, and Tongyang Li for helpful discussions. A.S.\ was supported in part by a Microsoft Research Faculty Fellowship, NSF CAREER Award CCF-1844855, NSF Grant CCF-1955039, a PayPal research award, and a Sloan Research Fellowship. C.Z.\ was supported in part by the Shoucheng Zhang Graduate Fellowship.

%% file: main.bbl
\providecommand{\bysame}{\leavevmode\hbox to3em{\hrulefill}\thinspace}
\begin{thebibliography}{10}

\bibitem{agarwal2017finding}
Naman Agarwal, Zeyuan Allen-Zhu, Brian Bullins, Elad Hazan, and Tengyu Ma,
  \emph{Finding approximate local minima faster than gradient descent},
  Proceedings of the 49th Annual ACM SIGACT Symposium on Theory of Computing,
  pp.~1195--1199, 2017,
  \mbox{\href{http://arxiv.org/abs/arXiv:1611.01146}{arXiv:1611.01146}}.

\bibitem{allen2017katyusha}
Zeyuan Allen-Zhu, \emph{Katyusha: The first direct acceleration of stochastic
  gradient methods}, The Journal of Machine Learning Research \textbf{18}
  (2017), no.~1, 8194--8244,
  \mbox{\href{http://arxiv.org/abs/arXiv:1603.05953}{arXiv:1603.05953}}.

\bibitem{allen2018neon2}
Zeyuan Allen-Zhu and Yuanzhi Li, \emph{Neon2: Finding local minima via
  first-order oracles}, Advances in Neural Information Processing Systems,
  pp.~3716--3726, 2018,
  \mbox{\href{http://arxiv.org/abs/arXiv:1711.06673}{arXiv:1711.06673}}.

\bibitem{an2021quantum}
Dong An, Noah Linden, Jin-Peng Liu, Ashley Montanaro, Changpeng Shao, and Jiasu
  Wang, \emph{Quantum-accelerated multilevel {Monte Carlo} methods for
  stochastic differential equations in mathematical finance}, Quantum
  \textbf{5} (2021), 481,
  \mbox{\href{http://arxiv.org/abs/arXiv:2012.06283}{arXiv:2012.06283}}.

\bibitem{vanApeldoorn2020optimization}
Joran~van Apeldoorn, Andr{\'a}s Gily{\'e}n, Sander Gribling, and Ronald
  de~Wolf, \emph{Convex optimization using quantum oracles}, Quantum \textbf{4}
  (2020), 220,
  \mbox{\href{http://arxiv.org/abs/arXiv:1809.00643}{arXiv:1809.00643}}.

\bibitem{arjevani2022lower}
Yossi Arjevani, Yair Carmon, John~C. Duchi, Dylan~J. Foster, Nathan Srebro, and
  Blake Woodworth, \emph{Lower bounds for non-convex stochastic optimization},
  Mathematical Programming (2022), 1--50,
  \mbox{\href{http://arxiv.org/abs/arXiv:1912.02365}{arXiv:1912.02365}}.

\bibitem{asi2021stochastic}
Hilal Asi, Yair Carmon, Arun Jambulapati, Yujia Jin, and Aaron Sidford,
  \emph{Stochastic bias-reduced gradient methods}, Advances in Neural
  Information Processing Systems \textbf{34} (2021), 10810--10822,
  \mbox{\href{http://arxiv.org/abs/arXiv:2106.09481}{arXiv:2106.09481}}.

\bibitem{bhojanapalli2016global}
Srinadh Bhojanapalli, Behnam Neyshabur, and Nati Srebro, \emph{Global
  optimality of local search for low rank matrix recovery}, Proceedings of the
  30th International Conference on Neural Information Processing Systems,
  pp.~3880--3888, 2016,
  \mbox{\href{http://arxiv.org/abs/arXiv:1605.07221}{arXiv:1605.07221}}.

\bibitem{birgin2017worst}
Ernesto~G. Birgin, J.~L. Gardenghi, Jos{\'e}~Mario Mart{\'\i}nez,
  Sandra~Augusta Santos, and Ph.~L. Toint, \emph{Worst-case evaluation
  complexity for unconstrained nonlinear optimization using high-order
  regularized models}, Mathematical Programming \textbf{163} (2017), no.~1,
  359--368.

\bibitem{blanchet2015unbiased}
Jose~H. Blanchet and Peter~W. Glynn, \emph{Unbiased {Monte Carlo} for
  optimization and functions of expectations via multi-level randomization},
  2015 {Winter Simulation Conference (WSC)}, pp.~3656--3667, IEEE, 2015.

\bibitem{boyd2008stochastic}
Stephen Boyd and Almir Mutapcic, \emph{Stochastic subgradient methods}, Lecture
  Notes for EE364b, Stanford University (2008), 97.

\bibitem{brandao2017SDP}
Fernando~G.S.L. Brand{\~a}o, Amir Kalev, Tongyang Li, Cedric Yen-Yu Lin,
  Krysta~M. Svore, and Xiaodi Wu, \emph{Quantum {SDP} solvers: {L}arge
  speed-ups, optimality, and applications to quantum learning}, Proceedings of
  the 46th International Colloquium on Automata, Languages, and Programming,
  Leibniz International Proceedings in Informatics (LIPIcs), vol. 132,
  pp.~27:1--27:14, Schloss Dagstuhl--Leibniz-Zentrum fuer Informatik, 2019,
  \mbox{\href{http://arxiv.org/abs/arXiv:1710.02581}{arXiv:1710.02581}}.

\bibitem{brandao2017quantum}
Fernando~G.S.L. Brand{\~a}o and Krysta~M. Svore, \emph{Quantum speed-ups for
  solving semidefinite programs}, 2017 IEEE 58th Annual Symposium on
  Foundations of Computer Science (FOCS), pp.~415--426, IEEE, 2017,
  \mbox{\href{http://arxiv.org/abs/1609.05537}{1609.05537}}.

\bibitem{brassard2002quantum}
Gilles Brassard, Peter H{\o}yer, Michele Mosca, and Alain Tapp, \emph{Quantum
  amplitude amplification and estimation}, Contemporary Mathematics
  \textbf{305} (2002), 53--74,
  \mbox{\href{http://arxiv.org/abs/arXiv:quant-ph/0005055}{arXiv:quant-ph/0005055}}.

\bibitem{bubeck2019complexity}
S{\'e}bastien Bubeck, Qijia Jiang, Yin-Tat Lee, Yuanzhi Li, and Aaron Sidford,
  \emph{Complexity of highly parallel non-smooth convex optimization}, Advances
  in neural information processing systems \textbf{32} (2019),
  \mbox{\href{http://arxiv.org/abs/arXiv:1906.10655}{arXiv:1906.10655}}.

\bibitem{carmon2017convex}
Yair Carmon, John~C. Duchi, Oliver Hinder, and Aaron Sidford, \emph{“{C}onvex
  until proven guilty”: Dimension-free acceleration of gradient descent on
  non-convex functions}, International conference on machine learning,
  pp.~654--663, PMLR, 2017,
  \mbox{\href{http://arxiv.org/abs/arXiv:1705.02766}{arXiv:1705.02766}}.

\bibitem{carmon2018accelerated}
Yair {Carmon}, John~C. Duchi, Oliver Hinder, and Aaron Sidford,
  \emph{Accelerated methods for nonconvex optimization}, SIAM Journal on
  Optimization \textbf{28} (2018), no.~2, 1751--1772,
  \mbox{\href{http://arxiv.org/abs/arXiv:1611.00756}{arXiv:1611.00756}}.

\bibitem{carmon2023resqueing}
Yair Carmon, Arun Jambulapati, Yujia Jin, Yin~Tat Lee, Daogao Liu, Aaron
  Sidford, and Kevin Tian, \emph{{ReSQueing} parallel and private stochastic
  convex optimization}, 2023,
  \mbox{\href{http://arxiv.org/abs/arXiv:2301.00457}{arXiv:2301.00457}}.

\bibitem{cesa2004generalization}
Nicolo Cesa-Bianchi, Alex Conconi, and Claudio Gentile, \emph{On the
  generalization ability of on-line learning algorithms}, IEEE Transactions on
  Information Theory \textbf{50} (2004), no.~9, 2050--2057.

\bibitem{chakrabarti2023quantum}
Shouvanik Chakrabarti, Andrew~M. Childs, Shih-Han Hung, Tongyang Li, Chunhao
  Wang, and Xiaodi Wu, \emph{Quantum algorithm for estimating volumes of convex
  bodies}, ACM Transactions on Quantum Computing \textbf{4} (2023), no.~3,
  1--60, \mbox{\href{http://arxiv.org/abs/arXiv:1908.03903}{arXiv:1908.03903}}.

\bibitem{chakrabarti2020optimization}
Shouvanik Chakrabarti, Andrew~M. Childs, Tongyang Li, and Xiaodi Wu,
  \emph{Quantum algorithms and lower bounds for convex optimization}, Quantum
  \textbf{4} (2020), 221,
  \mbox{\href{http://arxiv.org/abs/arXiv:1809.01731}{arXiv:1809.01731}}.

\bibitem{chewi2023complexity}
Sinho Chewi, S{\'e}bastien Bubeck, and Adil Salim, \emph{On the complexity of
  finding stationary points of smooth functions in one dimension},
  International Conference on Algorithmic Learning Theory, pp.~358--374, PMLR,
  2023,
  \mbox{\href{http://arxiv.org/abs/arXiv:arXiv:2209.07513}{arXiv:arXiv:2209.07513}}.

\bibitem{childs2022quantum}
Andrew~M. Childs, Jiaqi Leng, Tongyang Li, Jin-Peng Liu, and Chenyi Zhang,
  \emph{Quantum simulation of real-space dynamics}, Quantum \textbf{6} (2022),
  680, \mbox{\href{http://arxiv.org/abs/arXiv:2203.17006}{arXiv:2203.17006}}.

\bibitem{cornelissen2023sublinear}
Arjan Cornelissen and Yassine Hamoudi, \emph{A sublinear-time quantum algorithm
  for approximating partition functions}, Proceedings of the 2023 Annual
  ACM-SIAM Symposium on Discrete Algorithms (SODA), pp.~1245--1264, SIAM, 2023,
  \mbox{\href{http://arxiv.org/abs/arXiv:2207.08643}{arXiv:2207.08643}}.

\bibitem{cornelissen2022near}
Arjan Cornelissen, Yassine Hamoudi, and Sofiene Jerbi, \emph{Near-optimal
  quantum algorithms for multivariate mean estimation}, Proceedings of the 54th
  Annual ACM SIGACT Symposium on Theory of Computing, pp.~33--43, 2022,
  \mbox{\href{http://arxiv.org/abs/arXiv:2111.09787}{arXiv:2111.09787}}.

\bibitem{defazio2014saga}
Aaron Defazio, Francis Bach, and Simon Lacoste-Julien, \emph{Saga: A fast
  incremental gradient method with support for non-strongly convex composite
  objectives}, Advances in neural information processing systems \textbf{27}
  (2014).

\bibitem{duchi2018introductory}
John~C Duchi, \emph{Introductory lectures on stochastic optimization}, The
  mathematics of data \textbf{25} (2018), 99--186.

\bibitem{duchi2012randomized}
John~C. Duchi, Peter~L. Bartlett, and Martin~J. Wainwright, \emph{Randomized
  smoothing for stochastic optimization}, SIAM Journal on Optimization
  \textbf{22} (2012), no.~2, 674--701.

\bibitem{fang2018spider}
Cong Fang, Chris~Junchi Li, Zhouchen Lin, and Tong Zhang, \emph{{SPIDER}:
  Near-optimal non-convex optimization via stochastic path-integrated
  differential estimator}, Advances in Neural Information Processing Systems
  \textbf{31} (2018),
  \mbox{\href{http://arxiv.org/abs/arXiv:1807.01695}{arXiv:1807.01695}}.

\bibitem{fang2019sharp}
Cong Fang, Zhouchen Lin, and Tong Zhang, \emph{Sharp analysis for nonconvex
  {SGD} escaping from saddle points}, Conference on Learning Theory,
  pp.~1192--1234, 2019,
  \mbox{\href{http://arxiv.org/abs/arXiv:1902.00247}{arXiv:1902.00247}}.

\bibitem{garg2020no}
Ankit Garg, Robin Kothari, Praneeth Netrapalli, and Suhail Sherif, \emph{No
  quantum speedup over gradient descent for non-smooth convex optimization},
  2020, \mbox{\href{http://arxiv.org/abs/arXiv:2010.01801}{arXiv:2010.01801}}.

\bibitem{garg2021near}
Ankit {Garg}, Robin Kothari, Praneeth Netrapalli, and Suhail Sherif,
  \emph{Near-optimal lower bounds for convex optimization for all orders of
  smoothness}, Advances in Neural Information Processing Systems \textbf{34}
  (2021), 29874--29884,
  \mbox{\href{http://arxiv.org/abs/arXiv:2112.01118}{arXiv:2112.01118}}.

\bibitem{gasnikov2018global}
A.~V. Gasnikov, E.~A. Gorbunov, D.~A. Kovalev, A.~A. Mohammed, E.~O.
  Chernousova, et~al., \emph{The global rate of convergence for optimal tensor
  methods in smooth convex optimization}, Computer research and modeling
  \textbf{10} (2018), no.~6, 737--753,
  \mbox{\href{http://arxiv.org/abs/arXiv:1809.00382}{arXiv:1809.00382}}.

\bibitem{ge2015escaping}
Rong Ge, Furong Huang, Chi Jin, and Yang Yuan, \emph{Escaping from saddle
  points -- online stochastic gradient for tensor decomposition}, Proceedings
  of the 28th Conference on Learning Theory, Proceedings of Machine Learning
  Research, vol.~40, pp.~797--842, 2015,
  \mbox{\href{http://arxiv.org/abs/arXiv:1503.02101}{arXiv:1503.02101}}.

\bibitem{ge2017no}
Rong Ge, Chi Jin, and Yi~Zheng, \emph{No spurious local minima in nonconvex low
  rank problems: {A} unified geometric analysis}, International Conference on
  Machine Learning, pp.~1233--1242, PMLR, 2017,
  \mbox{\href{http://arxiv.org/abs/arXiv:1704.00708}{arXiv:1704.00708}}.

\bibitem{ge2016matrix}
Rong Ge, Jason~D. Lee, and Tengyu Ma, \emph{Matrix completion has no spurious
  local minimum}, Advances in Neural Information Processing Systems,
  pp.~2981--2989, 2016,
  \mbox{\href{http://arxiv.org/abs/arXiv:1605.07272}{arXiv:1605.07272}}.

\bibitem{ge2018learning}
Rong Ge, Jason~D. Lee, and Tengyu Ma, \emph{Learning one-hidden-layer neural
  networks with landscape design}, International Conference on Learning
  Representations, 2018,
  \mbox{\href{http://arxiv.org/abs/arXiv:1711.00501}{arXiv:1711.00501}}.

\bibitem{ge2019stabilized}
Rong Ge, Zhize Li, Weiyao Wang, and Xiang Wang, \emph{Stabilized {SVRG}:
  {S}imple variance reduction for nonconvex optimization}, Conference on
  Learning Theory, pp.~1394--1448, PMLR, 2019,
  \mbox{\href{http://arxiv.org/abs/arXiv:1905.00529}{arXiv:1905.00529}}.

\bibitem{ge2017optimization}
Rong Ge and Tengyu Ma, \emph{On the optimization landscape of tensor
  decompositions}, Advances in Neural Information Processing Systems,
  pp.~3656--3666, Curran Associates Inc., 2017,
  \mbox{\href{http://arxiv.org/abs/arXiv:1706.05598}{arXiv:1706.05598}}.

\bibitem{ghadimi2013stochastic}
Saeed Ghadimi and Guanghui Lan, \emph{Stochastic first-and zeroth-order methods
  for nonconvex stochastic programming}, SIAM Journal on Optimization
  \textbf{23} (2013), no.~4, 2341--2368.

\bibitem{giles2015multilevel}
Michael~B. Giles, \emph{Multilevel {Monte Carlo} methods}, Acta Numerica
  \textbf{24} (2015), 259--328.

\bibitem{gilyen2020distributional}
Andr{\'a}s Gily{\'e}n and Tongyang Li, \emph{Distributional property testing in
  a quantum world}, 11th Innovations in Theoretical Computer Science Conference
  (ITCS 2020), Schloss Dagstuhl-Leibniz-Zentrum f{\"u}r Informatik, 2020,
  \mbox{\href{http://arxiv.org/abs/arXiv:1902.00814}{arXiv:1902.00814}}.

\bibitem{gong2022robustness}
Weiyuan Gong, Chenyi Zhang, and Tongyang Li, \emph{Robustness of quantum
  algorithms for nonconvex optimization}, 2022,
  \mbox{\href{http://arxiv.org/abs/arXiv:2212.02548}{arXiv:2212.02548}}.

\bibitem{hamoudi2021quantum}
Yassine Hamoudi, \emph{Quantum {sub-Gaussian} mean estimator}, 29th Annual
  European Symposium on Algorithms (ESA 2021), Schloss Dagstuhl-Leibniz-Zentrum
  f{\"u}r Informatik, 2021,
  \mbox{\href{http://arxiv.org/abs/arXiv:2108.12172}{arXiv:2108.12172}}.

\bibitem{hamoudi2019quantum}
Yassine Hamoudi and Fr{\'e}d{\'e}ric Magniez, \emph{Quantum {Chebyshev's}
  inequality and applications}, 46th International Colloquium on Automata,
  Languages, and Programming (ICALP 2019), 2019,
  \mbox{\href{http://arxiv.org/abs/arXiv:1807.06456}{arXiv:1807.06456}}.

\bibitem{hardt2018gradient}
Moritz Hardt, Tengyu Ma, and Benjamin Recht, \emph{Gradient descent learns
  linear dynamical systems}, Journal of Machine Learning Research \textbf{19}
  (2018), no.~29, 1--44,
  \mbox{\href{http://arxiv.org/abs/arXiv:1609.05191}{arXiv:1609.05191}}.

\bibitem{hastie2009elements}
Trevor Hastie, Robert Tibshirani, Jerome~H. Friedman, and Jerome~H. Friedman,
  \emph{The elements of statistical learning: {Data} mining, inference, and
  prediction}, vol.~2, Springer, 2009.

\bibitem{jiang2020improved}
Haotian Jiang, Yin~Tat Lee, Zhao Song, and Sam Chiu-wai Wong, \emph{An improved
  cutting plane method for convex optimization, convex-concave games, and its
  applications}, Proceedings of the 52nd Annual ACM SIGACT Symposium on Theory
  of Computing, pp.~944--953, 2020,
  \mbox{\href{http://arxiv.org/abs/arXiv:2004.04250}{arXiv:2004.04250}}.

\bibitem{jin2017escape}
Chi Jin, Rong Ge, Praneeth Netrapalli, Sham~M. Kakade, and Michael~I. Jordan,
  \emph{How to escape saddle points efficiently}, Proceedings of the 34th
  International Conference on Machine Learning, vol.~70, pp.~1724--1732, 2017,
  \mbox{\href{http://arxiv.org/abs/arXiv:1703.00887}{arXiv:1703.00887}}.

\bibitem{jin2019stochastic}
Chi Jin, Praneeth Netrapalli, Rong Ge, Sham~M Kakade, and Michael~I. Jordan,
  \emph{On nonconvex optimization for machine learning: Gradients,
  stochasticity, and saddle points}, Journal of the ACM (JACM) \textbf{68}
  (2021), no.~2, 1--29,
  \mbox{\href{http://arxiv.org/abs/arXiv:1902.04811}{arXiv:1902.04811}}.

\bibitem{johnson2013accelerating}
Rie Johnson and Tong Zhang, \emph{Accelerating stochastic gradient descent
  using predictive variance reduction}, Advances in neural information
  processing systems \textbf{26} (2013).

\bibitem{jordan2005fast}
Stephen~P. Jordan, \emph{Fast quantum algorithm for numerical gradient
  estimation}, Physical Review Letters \textbf{95} (2005), no.~5, 050501,
  \mbox{\href{http://arxiv.org/abs/arXiv:quant-ph/0405146}{arXiv:quant-ph/0405146}}.

\bibitem{kelner2023semi}
Jonathan Kelner, Jerry Li, Allen~X. Liu, Aaron Sidford, and Kevin Tian,
  \emph{Semi-random sparse recovery in nearly-linear time}, The Thirty Sixth
  Annual Conference on Learning Theory, pp.~2352--2398, PMLR, 2023,
  \mbox{\href{http://arxiv.org/abs/arXiv:2203.04002}{arXiv:2203.04002}}.

\bibitem{kerenidis2018interior}
Iordanis Kerenidis and Anupam Prakash, \emph{A quantum interior point method
  for {LPs and SDPs}}, 2018,
  \mbox{\href{http://arxiv.org/abs/arXiv:1808.09266}{arXiv:1808.09266}}.

\bibitem{kitaev1995quantum}
A.~Yu. Kitaev, \emph{Quantum measurements and the {Abelian} stabilizer
  problem},  (1995),
  \mbox{\href{http://arxiv.org/abs/arXiv:quant-ph/9511026}{arXiv:quant-ph/9511026}}.

\bibitem{kothari2023mean}
Robin Kothari and Ryan O'Donnell, \emph{Mean estimation when you have the
  source code; or, quantum {Monte Carlo} methods}, Proceedings of the 2023
  Annual ACM-SIAM Symposium on Discrete Algorithms (SODA), pp.~1186--1215,
  SIAM, 2023,
  \mbox{\href{http://arxiv.org/abs/arXiv:2208.07544}{arXiv:2208.07544}}.

\bibitem{lan2012optimal}
Guanghui Lan, \emph{An optimal method for stochastic composite optimization},
  Mathematical Programming \textbf{133} (2012), no.~1-2, 365--397.

\bibitem{lee2015faster}
Yin~Tat Lee, Aaron Sidford, and Sam Chiu-wai Wong, \emph{A faster cutting plane
  method and its implications for combinatorial and convex optimization}, 2015
  IEEE 56th Annual Symposium on Foundations of Computer Science,
  pp.~1049--1065, IEEE, 2015,
  \mbox{\href{http://arxiv.org/abs/arXiv:1508.04874}{arXiv:1508.04874}}.

\bibitem{li2018quantum}
Tongyang Li and Xiaodi Wu, \emph{Quantum query complexity of entropy
  estimation}, IEEE Transactions on Information Theory \textbf{65} (2018),
  no.~5, 2899--2921,
  \mbox{\href{http://arxiv.org/abs/arXiv:1710.06025}{arXiv:1710.06025}}.

\bibitem{li2022enabling}
Xiantao Li, \emph{Enabling quantum speedup of markov chains using a multi-level
  approach}, 2022.

\bibitem{li2019ssrgd}
Zhize Li, \emph{{SSRGD}: {S}imple stochastic recursive gradient descent for
  escaping saddle points}, Advances in Neural Information Processing Systems
  \textbf{32} (2019), 1523--1533,
  \mbox{\href{http://arxiv.org/abs/arXiv:1904.09265}{arXiv:1904.09265}}.

\bibitem{liu2018adaptive}
Mingrui Liu, Zhe Li, Xiaoyu Wang, Jinfeng Yi, and Tianbao Yang, \emph{Adaptive
  negative curvature descent with applications in non-convex optimization},
  Advances in Neural Information Processing Systems \textbf{31} (2018),
  4858--4867.

\bibitem{liu2023quantum}
Yizhou Liu, Weijie~J. Su, and Tongyang Li, \emph{On quantum speedups for
  nonconvex optimization via quantum tunneling walks}, Quantum \textbf{7}
  (2023), 1030,
  \mbox{\href{http://arxiv.org/abs/arXiv:2209.14501}{arXiv:2209.14501}}.

\bibitem{lugosi2019mean}
G{\'a}bor Lugosi and Shahar Mendelson, \emph{Mean estimation and regression
  under heavy-tailed distributions: A survey}, Foundations of Computational
  Mathematics \textbf{19} (2019), no.~5, 1145--1190.

\bibitem{montanaro2015quantum}
Ashley Montanaro, \emph{Quantum speedup of {M}onte {C}arlo methods},
  Proceedings of the Royal Society A \textbf{471} (2015), no.~2181, 20150301,
  \mbox{\href{http://arxiv.org/abs/arXiv:1504.06987}{arXiv:1504.06987}}.

\bibitem{murty1985some}
Katta~G. Murty and Santosh~N. Kabadi, \emph{Some {NP-complete} problems in
  quadratic and nonlinear programming}, Mathematical Programming: Series A and
  B \textbf{39} (1987), no.~2, 117--129.

\bibitem{nemirovski1994efficient}
Arkadi Nemirovski, \emph{Efficient methods in convex programming}, Lecture
  notes (1994).

\bibitem{nemirovski2009robust}
Arkadi Nemirovski, Anatoli Juditsky, Guanghui Lan, and Alexander Shapiro,
  \emph{Robust stochastic approximation approach to stochastic programming},
  SIAM Journal on optimization \textbf{19} (2009), no.~4, 1574--1609.

\bibitem{nemirovskij1983problem}
Arkadij~Semenovi{\v{c}} Nemirovskij and David~Borisovich Yudin, \emph{Problem
  complexity and method efficiency in optimization}, 1983.

\bibitem{nesterov2003introductory}
Yurii Nesterov, \emph{Introductory lectures on convex optimization: A basic
  course}, vol.~87, Springer Science \& Business Media, 2003.

\bibitem{nesterov2006cubic}
Yurii Nesterov and Boris~T. Polyak, \emph{Cubic regularization of {N}ewton
  method and its global performance}, Mathematical Programming \textbf{108}
  (2006), no.~1, 177--205.

\bibitem{polyak1987introduction}
Boris~T. Polyak, \emph{Introduction to optimization. {Optimization} software},
  Inc., Publications Division, New York \textbf{1} (1987), 32.

\bibitem{shalev2007online}
Shai Shalev-Shwartz, \emph{Online learning: Theory, algorithms, and
  applications}, Hebrew University, 2007.

\bibitem{shalev2014understanding}
Shai Shalev-Shwartz and Shai Ben-David, \emph{Understanding machine learning:
  From theory to algorithms}, Cambridge university press, 2014.

\bibitem{van2019improvements}
Joran van Apeldoorn and Andr{\'a}s Gily{\'e}n, \emph{Improvements in quantum
  {SDP}-solving with applications}, 46th International Colloquium on Automata,
  Languages, and Programming (ICALP 2019), Schloss Dagstuhl-Leibniz-Zentrum
  fuer Informatik, 2019,
  \mbox{\href{http://arxiv.org/abs/1804.05058}{1804.05058}}.

\bibitem{van2020quantum}
Joran van Apeldoorn, Andr{\'a}s Gily{\'e}n, Sander Gribling, and Ronald
  de~Wolf, \emph{Quantum {SDP}-solvers: Better upper and lower bounds}, Quantum
  \textbf{4} (2020), 230,
  \mbox{\href{http://arxiv.org/abs/1705.01843}{1705.01843}}.

\bibitem{woodworth2017lower}
Blake Woodworth and Nathan Srebro, \emph{Lower bound for randomized first order
  convex optimization}, 2017,
  \mbox{\href{http://arxiv.org/abs/arXiv:1709.03594}{arXiv:1709.03594}}.

\bibitem{xu2017neon}
Yi~Xu, Rong Jin, and Tianbao Yang, \emph{{NEON}+: Accelerated gradient methods
  for extracting negative curvature for non-convex optimization}, 2017,
  \mbox{\href{http://arxiv.org/abs/arXiv:1712.01033}{arXiv:1712.01033}}.

\bibitem{yu2018third}
Yaodong Yu, Pan Xu, and Quanquan Gu, \emph{Third-order smoothness helps:
  {F}aster stochastic optimization algorithms for finding local minima},
  Advances in Neural Information Processing Systems (2018), 4530--4540.

\bibitem{zhang2021quantum}
Chenyi Zhang, Jiaqi Leng, and Tongyang Li, \emph{Quantum algorithms for
  escaping from saddle points}, Quantum \textbf{5} (2021), 529,
  \mbox{\href{http://arxiv.org/abs/arXiv:2007.10253}{arXiv:2007.10253}}.

\bibitem{zhang2021escape}
Chenyi Zhang and Tongyang Li, \emph{Escape saddle points by a simple
  gradient-descent based algorithm}, Advances in Neural Information Processing
  Systems \textbf{34} (2021), 8545--8556,
  \mbox{\href{http://arxiv.org/abs/arXiv:2111.14069}{arXiv:2111.14069}}.

\bibitem{zhang2022quantum}
Chenyi {Zhang} and Tongyang Li, \emph{Quantum lower bounds for finding
  stationary points of nonconvex functions}, 2022,
  \mbox{\href{http://arxiv.org/abs/arXiv:2212.03906}{arXiv:2212.03906}}.

\bibitem{zhang2018primal}
Xiao Zhang, Lingxiao Wang, Yaodong Yu, and Quanquan Gu, \emph{A primal-dual
  analysis of global optimality in nonconvex low-rank matrix recovery},
  International Conference on Machine Learning, pp.~5862--5871, PMLR, 2018.

\bibitem{zhou2018stochastic}
Dongruo Zhou, Pan Xu, and Quanquan Gu, \emph{Stochastic variance-reduced cubic
  regularized {N}ewton methods}, International Conference on Machine Learning,
  pp.~5990--5999, PMLR, 2018,
  \mbox{\href{http://arxiv.org/abs/arXiv:1802.04796}{arXiv:1802.04796}}.

\bibitem{zhou2020stochastic}
Dongruo {Zhou}, Pan Xu, and Quanquan Gu, \emph{Stochastic nested variance
  reduction for nonconvex optimization}, Journal of Machine Learning Research
  \textbf{21} (2020), no.~103, 1--63,
  \mbox{\href{http://arxiv.org/abs/arXiv:1806.07811}{arXiv:1806.07811}}.

\end{thebibliography}
